\documentclass{article}

\usepackage[OT1]{fontenc}

\usepackage{amsmath,amsfonts}
\usepackage{graphicx}

\usepackage[english]{babel}
\usepackage [autostyle, english = american]{csquotes}
    \MakeOuterQuote{"}
\usepackage{tikz,tikz-cd}
\usetikzlibrary{decorations.pathmorphing}
\usepackage{natbib}
\usepackage{amssymb,mathtools}
\usepackage{hyperref}
\hypersetup{
    colorlinks=true,
    linkcolor=blue,
    filecolor=blue,      
    urlcolor=teal,
    citecolor=purple,
    pdftitle={Groupoidal Realizability for Intensional Type Theory},
    pdfauthor = {Sam Speight}
    }
\usepackage{cleveref}
\usepackage{adjustbox}
\usepackage{enumitem}
\usepackage{faktor}
\usepackage{stmaryrd}

\usepackage{amsthm}
    \theoremstyle{plain}
\newtheorem{thm}{Theorem}
\newtheorem{lem}[thm]{Lemma}
\newtheorem{prop}[thm]{Proposition}
\newtheorem{defn}[thm]{Definition}

    \theoremstyle{remark}
\newtheorem{rem}[thm]{Remark}
\newtheorem{example}[thm]{Example}

\crefname{thm}{theorem}{theorems}
\Crefname{thm}{Theorem}{Theorems}

\title{Groupoidal Realizability \\ for Intensional Type Theory}
\author{Samuel Luke Speight}
\date{}

\begin{document}
  
\maketitle

\begin{abstract}
We develop realizability models of intensional type theory, based on groupoids, wherein realizers themselves carry non-trivial (non-discrete) homotopical structure. In the spirit of realizability, this is intended to formalize a homotopical BHK interpretation, whereby evidence for an identification is a path. Specifically, we study partitioned groupoidal assemblies. Categories of such are parameterised by "realizer categories" (instead of the usual partial combinatory algebras) that come equipped with an interval \textit{qua} internal cogroupoid. The interval furnishes a notion of homotopy as well as a fundamental groupoid construction. Objects in a base groupoid are realized by points in the fundamental groupoid of some object from the realizer category; isomorphisms in the base groupoid are realized by paths in said fundamental groupoid. The main result is that, under mild conditions on the realizer category, the ensuing category of partitioned groupoidal assemblies models intensional (1-truncated) type theory without function extensionality. Moreover, when the underlying realizer category is "untyped", there exists an impredicative universe of 1-types (the modest fibrations). This is a groupoidal analogue of the traditional situation.
\end{abstract}

{\small \tableofcontents}

%%%%%

\section{Introduction}\label{sec:introduction}

It is often said that realizability formalises the BHK interpretation \citep[see][for example]{bauer12}. The latter is an informal interpretation, inductively prescribing what counts as evidence for logical formulas \citep{heyting30,heyting31,kolmogorov32,heyting34,heyting56}:
\begin{itemize}
    \item Evidence for $\phi\land\psi$ consists of evidence for $\phi$ and evidence for $\psi$.
    \item Evidence for $\phi\lor\psi$ is an identifier for either $\phi$ or $\psi$ together with evidence for the formula identified.
    \item Evidence for $\phi\rightarrow\psi$ is a process (or method, or function) that, when given evidence for $\phi$, produces evidence for $\psi$.
    \item Evidence for $\exists x\in A.\, \phi(x)$ consists of a representation of some object $a\in A$ and evidence for $\phi(a)$.
    \item Evidence for $\forall x\in A.\, \phi(x)$ is again a process that, when given a representation of any object $a\in A$, produces evidence for $\phi(a)$.
\end{itemize}

In a realizability interpretation, formally specified realizers play the role of evidence. Moreover, informal terms like "process" are given a precise, mathematical meaning. For example, in the original realizability interpretation, Kleene's (\citeyear{kleene45}) so-called "number realizability" interpretation of Heyting arithmetic, realizers are natural numbers. Processes are also natural numbers: Recall that partial computable functions may be effectively enumerated. Given a natural number $m$, the process $n$ returns the result---if it is defined---of applying the partial computable function encoded by $n$ to the argument $m$ (often denoted either by $\varphi_n(m)$ or by $\{n\}(m)$). Insofar as realizers are computational in nature, realizability interpretations unveil computational content of logical or mathematical theories.

A few decades after Kleene's number realizability, \citet{hyland82} introduced the \textit{effective topos} $\mathbf{Eff}$. This is a universe for computable mathematics: its internal logic extends number realizability. More generally, a realizability topos can be constructed over an arbitrary partial combinatory algebra (PCA); $\mathbf{Eff}$ is the realizability topos built over the PCA $\mathcal{K}_1$ (known as "Kleene's first algebra") of natural numbers and partial computable functions. The resulting "realizability topos" is a universe for computable mathematics, where the notion of computation is determined by the PCA.

The realizability topos $\mathbf{RT}(\mathcal{A})$ over an arbitrary PCA $\mathcal{A}$ boasts a number of subcategories that model powerful and expressive type theories. Each of these categories has an elementary construction directly over the PCA, independently of the realizability topos. There is the category $\mathbf{Asm}(\mathcal{A})$ of \textit{assemblies}, which is regular and locally cartesian closed. The full subcategory $\mathbf{PAsm}(\mathcal{A}) \subseteq \mathbf{Asm}(\mathcal{A})$ spanned by the \textit{partitioned assemblies} is \textit{weakly} locally cartesian closed and therefore a model of type theory without function extensionality. Note that $\mathbf{RT}(\mathcal{A})$ is exact, and can be obtained as the ex/lex completion of $\mathbf{PAsm}(\mathcal{A})$ or the ex/reg completion of $\mathbf{Asm}(\mathcal{A})$ \citep{robinson90}.
\begin{align*}
    \begin{tikzcd}[ampersand replacement=\&]
        \mathbf{PAsm}(\mathcal{A})
        \arrow[rr, hookrightarrow]
        \arrow[rr, bend left=15, rightsquigarrow, "{\text{reg/lex}}" description]
        \arrow[rrrr, bend left=30, rightsquigarrow, "{\text{ex/lex}}" description]
        \& \&
        \mathbf{Asm}(\mathcal{A})
        \arrow[rr, hookrightarrow]
        \arrow[rr, bend left=15, rightsquigarrow, "{\text{ex/reg}}" description]
        \& \&
        \mathbf{RT}(\mathcal{A})
    \end{tikzcd}
\end{align*}

Both $\mathbf{Asm}(\mathcal{A})$ and $\mathbf{PAsm}(\mathcal{A})$ contain an \textit{impredicative} universe, that is, one closed under large products. The small (closed) types are the \textit{modest} (partitioned) assemblies.

The realizability categories discussed above model some version of \textit{extensional} type theory. In number realizability, any natural number realizes a true equation; thus realizers of equations give no information beyond the fact that the equation holds. The so-called "homotopy interpretation" of type theory understands terms of the identity type $\mathsf{Id}_A(a,b)$ as paths between the points $a$ and $b$ in the space $A$ \citep{hottbook13}. This can be reformulated as a BHK clause:
\begin{itemize}
    \item Evidence for an identification of $a,b\in A$ is a path between representations of $a$ and $b$ in the space representing $A$.
\end{itemize}
The terms "path" and "space" are meant to be understood informally, just like terms such as "process" in the usual BHK clauses. The informal interpretation obtained by adding the above clause to the usual BHK interpretation shall be called the "homotopy BHK interpretation".

\begin{quote}
    The overarching aim of this paper is to develop realizability models of \textit{intensional} type theory that formalize the \textit{homotopy BHK interpretation} just as the realizability models of extensional type theory discussed above formalize the usual BHK interpretation.
\end{quote}
To this end, we shall equip Hofmann and Streicher's (\citeyear{hofmann98}) groupoid model of type theory with a suitable notion of realizability, just as traditional realizability categories do with the set model.

Instead of constructing realizability models over PCAs, we will do so over structured "realizer" categories. This idea can be traced back to \citet{lambek95} and \citet{abramsky95}: over an arbitrary category, the former constructs a category of partial equivalence relations (PERs), while the latter constructs categories of assemblies and modest sets. Both proceed to show that the ensuing realizability categories inherit structure from the initial realizer category, and in some cases even improve a weak universal property to a full one, as is a theme in realizability. Later, analyses of realizability using realizer categories were given by \citet{birkedal99,birkedal00a,birkedal00b} and \citet{robinson01}. The former constructs models of type theory in assemblies and modest sets over weakly closed partial cartesian categories (ie. categories with a notion of partial map that are weakly cartesian closed in a suitable sense) via a tripos-theoretic approach. The latter studies necessary conditions on a realizer category for the construction of realizability models with certain type-theoretic features. Both consider when a topos may be obtained.

In order to formalize the homotopical BHK interpretation, realizers need carry non-trivial (non-discrete) homotopical structure. In the current approach, the key to attaining this is to take part of the structure of a realizer category $\mathbf{R}$ to be an interval \textit{qua} internal cogroupoid $\mathbb{I}\in\mathbf{R}$. The interval supplies a notion of homotopy internal to $\mathbf{R}$ as well as a fundamental groupoid construction $\Pi:\mathbf{R}\rightarrow\mathbf{Gpd}$. Objects in a base groupoid $X$ are realized by points in the fundamental groupoid of some object $A\in\mathbf{R}$, and isomorphisms in $X$ are realized by paths in $\Pi A$. In this paper, we concentrate on \textit{partitioned groupoidal assemblies}, a groupoidal analogue of partitioned assemblies.

\subsection{Outline}\label{sec:outline}

The next section (\ref{sec:relatedwork}) discusses related work. Following this, \Cref{sec:setbasedrealizability} briefly reviews set-based realizability over combinatory algebras (\ref{sec:combinatoryalgebras}), typed combinatory algebras (\ref{sec:typedcombinatoryalgebras}) and categories (\ref{sec:categories}). \Cref{sec:realizercategories} discusses realizer categories for groupoidal realizability. These come equipped with an interval, as described in \Cref{sec:intervals}, which facilitaties a notion of homotopy (\Cref{sec:homotopies}) as well as a fundamental groupoid construction (\Cref{sec:fundamentalgroupoids}). Untyped realizer categories (required for impredicative universes) are discussed in \Cref{sec:untypedrealizercategories}.

\Cref{sec:categoriesofpartitionedgroupoidalassemblies} introduces the main players: partitioned groupoidal assemblies. The category $\mathbf{PGAsm}(\mathbf{R},\mathbb{I})$ of partitioned groupoidal assemblies over the realizer category $(\mathbf{R},\mathbb{I})$ is first studied as a (2,1)-category (\Cref{sec:asa2category}) and then as a path category (\Cref{sec:asapathcategory}). Dependent products are exhibited in \Cref{sec:dependentproducts}. \Cref{sec:impredicativeuniversesofmodestfibrations} is devoted to impredicative universes of modest fibrations. \Cref{sec:outlook} summarizes and outlines future work.

\subsection{Related work}\label{sec:relatedwork}

As far as we understand, the first to consider realizability semantics for intensional type theory were \citet{hofstra13}. They equip the syntax of 1-truncated intensional type theory with a notion of realizability allowing them to show that the syntactic groupoid associated to the type theory generated by a graph has the same homotopy type as the free groupoid on this graph.

Further work at the intersection of realizability and intensional type theory was motivated by the search for \textit{impredicative and univalent} universes of (homotopically) higher types. \citet{uemura18} gave a model with such a universe in the category of "cubical assemblies", ie. cubical objects internal to the model of extensional type theory in $\mathbf{Asm}(\mathcal{K}_1)$. Another type-theoretic principle of interest is \textit{propositional resizing}, which states that every proposition is equivalent to a "small" one (one living in the lowest universe). Uemura's model does not satisfy propositional resizing, exhibiting this as a distinct form of impredicativity. \citet{swan21} show that univalence is consistent with Church's thesis (CT, all functions on the natural numbers are computable): though CT does not hold in the cubical assemblies model, there is a reflective submodel in which it does.

In contrast to cubical assemblies, \citet{vandenberg18b} exhibits $\mathbf{Eff}$ as the homotopy category of a path category in which there is an impredicative and univalent universe of propositions that \textit{does} satisfy propositional resizing. A more complicated path category contains an impredicative and univalent universe of sets satisfying propositional resizing.

\citet{angiuli17} introduce \textit{computational higher-dimensional type theory}, based on a cubical generalization of Martin-Lof's meaning explanations \citep{angiuli17meaning}. This involves a realizability---in particular, a PER---construction over a cubical programming language, which produces a model of cubical type theory. This approach has been extended to univalent (though not impredicative) universes \citep{angiuli18} and higher inductive types \citep{cavallo19}. The realizers here, being expressions in a \emph{cubical} programming language, do carry non-trivial homotopical structure, making this approach closest in spirit to ours (though vastly different in technical detail).

Almost all of the novel results in this paper first appeared in the author's DPhil thesis \citep{speight23}.

\section{From combinatory algebras to realizer categories}\label{sec:fromcombinatoryalgebrastorealizercategories}

\subsection{Set-based realizability}\label{sec:setbasedrealizability}

In this section, we briefly review categories of assemblies over: first, (untyped) combinatory algebras; next, \textit{typed} combinatory algebras; and finally over categories with a terminal object. We show how each of these cases subsumes the previous case. Although the case of (partial) combinatory algebras is best known, it appears that categories are more easily adapted to the groupoidal setting. We put aside the matter of partiality from hereon in.

\subsubsection{Combinatory algebras}\label{sec:combinatoryalgebras}
A combinatory algebra (CA) $\mathcal{A}$ consists of a set $\mathcal{A}$ and a binary "application" operation
\begin{align*}
    (-)\cdot(?): \mathcal{A}\times\mathcal{A}\rightarrow\mathcal{A}
\end{align*}
such that there exist elements ("combinators") $\mathsf{k},\mathsf{s}\in\mathcal{A}$ satisfying $\mathsf{k}\cdot a \cdot b = a$ and $\mathsf{s} \cdot a \cdot b \cdot c = a \cdot c \cdot (b \cdot c)$. From now on, we adopt the convention of left associativity and may omit the application symbol.

CAs enjoy a property known as "combinatory completeness", which, roughly speaking, allows them to mimic the $\lambda$-calculus \citep{schonfinkel24,curry30}. A \textit{polynomial} over $\mathcal{A}$ is a formal expression built from the grammar:
\begin{align*}
    t \Coloneqq x\in\mathcal{V} \;|\; a\in\mathcal{A} \;|\; t \cdot t
\end{align*}
where $\mathcal{V}$ is a countably infinite set of variables. Clearly, we can talk about free variables of a polynomial and substitution $t[a/x]$ of the element $a\in\mathcal{A}$ for the free variable $x\in\mathcal{V}$ in the polynomial $t$. Moreover, closed polynomials have an obvious interpretation in $\mathcal{A}$. Combinatory completeness says that for every polynomial $t$ over $\mathcal{A}$ and every variable $x$ there exists a polynomial $\lambda x.\, t$ such that $\mathsf{FV}(\lambda x.\, t) \subseteq \mathsf{FV}(t) - \{x\}$ and for all $a\in\mathcal{A}$:
\begin{align*}
    (\lambda x.\, t) a = t[a/x]
\end{align*}

Indeed, an example of a CA is the set of $\lambda$-terms up to $\beta$-equivalence, where application is $\lambda$-application, $\mathsf{k}\coloneq \lambda x y.\, x$ and $\mathsf{s} = \lambda f g x.\, fx(gx)$. (Kleene's first algebra $\mathcal{K}_1$, mentioned in the previous section, is not a combinatory algebra but a \textit{partial} combinatory algebra: the application of one element to another may not be defined (Turing machines may not halt).)

The objects of the category $\mathbf{Asm}(\mathcal{A})$ of assemblies over the CA $\mathcal{A}$ are pairs $(X,\Vdash_X)$ consisting of a set $X$ and a "realizability" relation $\Vdash_X \subseteq \mathcal{A}\times X$ (written infix) such that $\forall x\in X. \, \exists a\in\mathcal{A}. \, a \Vdash_X x$. Assemblies are sometimes thought of as datatypes, whose "values" are implemented (by their realizers) in the programming language given by the CA. A morphism $f:(X,\Vdash_X) \rightarrow (Y,\Vdash_Y)$ is a set-function $f:X\rightarrow Y$ such that $\exists e\in\mathcal{A}.\, \forall x\in X. \, \forall a\in\mathcal{A}. \, a\Vdash_X x \rightarrow ea \Vdash_Y f(x)$. We say that the function $f$ is realized by $e$, and write $e\Vdash f$. Identities in $\mathbf{Asm}(\mathcal{A})$ are identity functions, which are realized by $\lambda x.\, x$. Composition is also inherited from $\mathbf{Set}$: if $e\Vdash f: X\rightarrow Y$ and $e'\Vdash g:Y \rightarrow Z$ then $\lambda x.\, e'(ex) \Vdash gf:X \rightarrow Z$.

An assembly $X$ is \textit{modest} when elements are uniquely determined by any of their realizers: $\forall x,x'\in X. \, \forall a\in\mathcal{A}. \, a \Vdash_X x \land a \Vdash_X x' \rightarrow x=x'$. An assembly $X$ is \textit{partitioned} when the relation $\Vdash_X$ is actually a function:
\begin{align*}
    \left\Vert - \right\Vert_X : X \rightarrow \mathcal{A}
\end{align*}
(so each element of $X$ has exactly one realizer). Of course, an assembly $X$ may be both modest and partitioned, in which case the function $\left\Vert - \right\Vert_X$ is an injection. Thus every modest partitioned assembly can be identified with a subset of $\mathcal{A}$: this is used in the construction of an impredicative universe of modest partitioned assemblies in $\mathbf{PAsm}(\mathcal{A})$.

\subsubsection{Typed combinatory algebras}\label{sec:typedcombinatoryalgebras}
\textit{Typed} (partial) combinatory algebras were introduced by \cite{longley99}. A typed combinatory algebra (TCA) is built over a "type system": a non-empty set $\mathcal{T}$ of "types" that is closed under the operation $\shortrightarrow$ (which associates to the right by convention). (Normally one would also consider products, but for present purposes we need not.) A TCA $\mathcal{A}$ over the type system $\mathcal{T}$ is a family
\begin{align*}
    (\mathcal{A}_A)_{A\in\mathcal{T}}
\end{align*}
of non-empty sets indexed by types, together with a family of typed application operations
\begin{align*}
    (-) \cdot_{A,B} (?): \mathcal{A}_{A\shortrightarrow B} \times \mathcal{A}_A \rightarrow \mathcal{A}_B
\end{align*}
(again we tend to omit the symbol) such that the following typed combinators are required to exist and satisfy the given equations.
\begin{align*}
    &\mathsf{k}_{A,B} \in \mathcal{A}_{A\shortrightarrow B \shortrightarrow A}
    &\mathsf{k}_{A,B} a b = a
    \\
    &\mathsf{s}_{A,B,C} \in \mathcal{A}_{(A \shortrightarrow B \shortrightarrow C) \shortrightarrow (A \shortrightarrow B) \shortrightarrow A \shortrightarrow C}
    &\mathsf{s}_{A,B,C} f g a = f a (g a)
\end{align*}
TCAs are combinatorially complete in a typed sense \citep[see][for details]{bauer22}.

If $\mathcal{A}$ is a TCA, then an object of $\mathbf{Asm}(\mathcal{A})$ is a triple $(X,A,\Vdash_X)$, where $X\in\mathbf{Set}$, $A\in\mathcal{T}$ is the "realizer type" and $\Vdash_X \subseteq \mathcal{A}_A \times X$ (so realizers of elements of $X$ all have the same type) is such that $\forall x\in X. \, \exists a\in\mathcal{A}_A. \, a \Vdash_X x$. A morphism $(X,A,\Vdash_X)\rightarrow (Y,B,\Vdash_Y)$ is a function $f:X\rightarrow Y$ such that $\exists e\in\mathcal{A}_{A\shortrightarrow B}.\, \forall x\in X. \, \forall a\in\mathcal{A}. \, a\Vdash_X x \rightarrow ea \Vdash_Y f(x)$.

Any CA can be regarded as a TCA with a single type $U=U\shortrightarrow U$ (in such a way that the ensuing categories of assemblies are isomorphic). Furthermore, any TCA can be equipped with a "unit" type such that the ensuing category of assemblies remains unchanged up to equivalence. That is, given any type system $\mathcal{T}$, we consider the augmented type system $\mathcal{T}^1 \coloneqq \mathcal{T} \cup \{ 1\}$, where $1\shortrightarrow A \coloneq A$ and $A \shortrightarrow 1 \coloneq 1$. Then we obtain the augmented TCA $\mathcal{A}^1$ over $\mathcal{T}_1$ by setting $\mathcal{A}^1_1 \coloneq 1$ (the terminal set). The augmented application operation is given by:
\begin{align*}
    &a \cdot_{1\shortrightarrow A} * \coloneq a
    &* \cdot_{A\shortrightarrow 1} a \coloneq *
\end{align*}
For the combinators, it suffices to specify the following (where $A,A'\in\mathcal{T}^1$).
\begin{align*}
    &\mathsf{k}_{1,A} \coloneq *
    &\mathsf{k}_{A,1} \coloneq \lambda x.\, x
    \\
    &\mathsf{s}_{1,A,A'} \coloneq \lambda f a.\, fa \eqcolon \mathsf{s}_{A,1,A'}
    &\mathsf{s}_{A,A',1} \coloneq *
\end{align*}
Here we use typed combinatorial completeness of $\mathcal{A}$. Checking the relevant equations is straightforward.

\begin{prop}
    Let $\mathcal{A}$ be a TCA. $\mathbf{Asm}(\mathcal{A}) \simeq \mathbf{Asm}(\mathcal{A}^1)$.
\end{prop}
\begin{proof}
    $\mathbf{Asm}(\mathcal{A})$ is a full subcategory of $\mathbf{Asm}(\mathcal{A}^1)$. For any $X=(X,1,\Vdash_X)\in\mathbf{Asm}(\mathcal{A}^1)$ we define an isomorphism
    \begin{align*}
        f: X \rightarrow X' \coloneq \left(X,A,\left\Vert-\right\Vert_{X'}\right)
    \end{align*}
    where for all $x\in X$ we set
    \begin{align*}
        \left\Vert x \right\Vert_{X'} \coloneq a_0
    \end{align*}
    for an arbitrarily chosen $a_0 \in \mathcal{A}_A$. The function $f\coloneq \mathsf{id}_X$ is realized by $a_0 \in \mathcal{A}_A = \mathcal{A}_{1\shortrightarrow A}$. The inverse $f^{-1} \coloneq \mathsf{id}_X$ is realized by $*\in\mathcal{A}_1 = \mathcal{A}_{A\shortrightarrow 1}$. 
\end{proof}

\subsubsection{Categories}\label{sec:categories}

A category, like a TCA, gives rise to an \textit{a priori} typed notion of realizability: the types are the objects of the category. Suppose $\mathbf{R}$ is a category with a terminal object $1$. Then we have the functor
\begin{align*}
    \Pi \coloneq \mathbf{R}(1,-): \mathbf{R} \rightarrow \mathbf{Set}
\end{align*}

The objects of the category $\mathbf{Asm}(\mathbf{R})$ of assemblies over the "realizer category" $\mathbf{R}$ are triples $(X,A,\Vdash_X)$, where $X\in\mathbf{Set}$, $A\in\mathbf{R}$ (a realizer type is now an object from the realizer category) and $\Vdash_X \subseteq \Pi A \times X$ is such that $\forall x\in X. \, \exists a\in \Pi A. \, a \Vdash_X x$. A morphism $(X,A,\Vdash_X)\rightarrow (Y,B,\Vdash_Y)$ is a function $f:X\rightarrow Y$ such that $\exists e:A\rightarrow B \in \mathbf{R}.\, \forall x\in X. \, \forall a\in\mathcal{A}. \, a\Vdash_X x \rightarrow \Pi(e)(a) = e\circ a \Vdash_Y f(x)$.

Given a TCA $\mathcal{A}$ with a unit type $1$ ($\mathcal{A}_1 = 1$) we may build a category $\mathbf{R}(\mathcal{A})$ with a terminal object. The objects of $\mathbf{R}(\mathcal{A})$ are the types of $\mathcal{A}$. A morphism $A \rightarrow B$ is a "computable" function $k: \mathcal{A}_A \rightarrow \mathcal{A}_B$. The function $k$ is computable iff there exists $e\in\mathcal{A}_{A\shortrightarrow B}$ representing it, that is, for all $a\in\mathcal{A}_A$ we have $k(a)=ea$. Note that computable functions $\mathcal{A}_1 \rightarrow \mathcal{A}_A$ are in bijective correspondence with elements of $\mathcal{A}_A$: a function $k:\mathcal{A}_1 \rightarrow \mathcal{A}_A$ is represented by $\lambda x.\, k(*)$.

\begin{prop}
    Let $\mathcal{A}$ be a TCA with a unit type. $\mathbf{Asm}(\mathcal{A}) \cong \mathbf{Asm}(\mathbf{R}(\mathcal{A}))$.
\end{prop}
\begin{proof}
    Up to the identification of computable functions $\mathcal{A}_1 \rightarrow \mathcal{A}_A$ with elements of $\mathcal{A}_A$, the functors going back and forth between these two categories are identities on both objects and arrows. A function $f$ realized by $e$ in $\mathbf{Asm}(\mathcal{A})$ is realized by the computable function $e\cdot(-)$ (represented by $e$) in $\mathbf{Asm}(\mathbf{R}(\mathcal{A}))$; conversely, a function $g$ realized by the computable function $k$ in $\mathbf{Asm}(\mathbf{R}(\mathcal{A}))$ is realized by $e_k$ in $\mathbf{Asm}(\mathcal{A})$, where $e_k$ represents $k$.
\end{proof}

\subsection{Realizer categories}\label{sec:realizercategories}

In this development, the key to obtaining realizers with non-trivial homotopical structure is to take part of the structure of a realizer category to be an interval \textit{qua} internal cogroupoid. We will assume that realizer categories are cartesian closed; this is a fairly mild assumption and provides a pleasant context in which to work with intervals. \citep{warren08,warren12} are excellent references when it comes to intervals \textit{qua} cogroupoids. Realizer categories come in \textit{typed} and \textit{untyped} varieties.

\subsubsection{Intervals}\label{sec:intervals}
Let $\mathbf{R}$ be a category with a terminal object. An interval (cogroupoid) $\mathbb{I}\in\mathbf{R}$ is a diagram of the form
\begin{align*}
    \begin{tikzcd}[ampersand replacement=\&]
        \mathbb{I}_0 = 1
        \ar[r, bend left=50, "0"]
        \ar[r, bend right=50, swap, "1"]
        \&
        \mathbb{I}_1
        \ar[r, bend left=50, "i_0"]
        \ar[r, bend right=50, swap, "i_1"]
        \ar[l, swap, "*"]
        \ar[r, "2"]
        \arrow[loop above, distance=3em, "\sigma"]
        \&
        \mathbb{I}_2
        \arrow[r, bend left=50, "j_0"]
        \arrow[r, bend right=50, swap, "j_1"]
        \&
        \mathbb{I}_3
    \end{tikzcd}
\end{align*}
We require that $\mathbb{I}_0=1$ is terminal. $\mathbb{I}_0$ and $\mathbb{I}_1$ are respectively known as the "object of coobjects" and the "object of coarrows". The diagram
\begin{align}\label{eqn:pushout}
    \begin{tikzcd}[ampersand replacement=\&]
        \mathbb{I}_0
        \arrow[r, "0"]
        \arrow[d, swap, "1"] 
        \&
        \mathbb{I}_1
        \arrow[d, "i_1"]
        \\
        \mathbb{I}_1
        \arrow[r, swap, "i_0"]
        \&
        \mathbb{I}_2
        \arrow[ul, phantom, "\ulcorner", very near start]
    \end{tikzcd}
\end{align}
is required to be a pushout. Maps $\mathbb{I}_1 \rightarrow A$ are thought of as paths in $A$, and so, the pushout allows us to concatenate two paths $\alpha,\beta:\mathbb{I}_1 \rightarrow A$ that match nose to tail: $\beta 0 = \alpha 1$. The result is $[\beta,\alpha]:\mathbb{I}_2 \rightarrow A$ (a path with twice the length of $\alpha$ and $\beta$). Likewise, the following is a pushout.
\begin{align*}
    \begin{tikzcd}[ampersand replacement=\&]
        \mathbb{I}_1
        \arrow[r, "i_1"]
        \arrow[d, swap, "i_0"] 
        \&
        \mathbb{I}_2 \arrow[d, "j_0"]
        \\
        \mathbb{I}_2
        \arrow[r, swap, "j_1"]
        \&
        \mathbb{I}_3
        \arrow[ul, phantom, "\ulcorner", very near start]
    \end{tikzcd}
\end{align*}
To round off the definition, the cogroupoid axioms are required to hold. The first set makes sure that the endpoint (or, respectively, cosource and cotarget) maps $0,1$ play nicely with cocomposition $2$ and coidentity $*$.
\begin{align*}
    \begin{tikzcd}[ampersand replacement=\&]
        \mathbb{I}_0
        \arrow[r, "0"]
        \arrow[d, swap, "0"] 
        \&
        \mathbb{I}_1
        \arrow[d, "2"]
        \\
        \mathbb{I}_1
        \arrow[r, swap, "i_0"]
        \&
        \mathbb{I}_2
    \end{tikzcd}
    \quad\quad\quad
    \begin{tikzcd}[ampersand replacement=\&]
        \mathbb{I}_0
        \arrow[r, "1"]
        \arrow[d, swap, "1"] 
        \&
        \mathbb{I}_1
        \arrow[d, "2"]
        \\
        \mathbb{I}_1
        \arrow[r, swap, "i_1"]
        \&
        \mathbb{I}_2
    \end{tikzcd}
    \quad\quad\quad
    \begin{tikzcd}[ampersand replacement=\&]
        \mathbb{I}_0
        \arrow[r, "0"]
        \arrow[equal, dr]
        \&
        \mathbb{I}_1
        \arrow[d, "*"]
        \&
        \mathbb{I}_0
        \arrow[l, swap, "1"]
        \arrow[equal, dl]
        \\
        \&
        \mathbb{I}_0
    \end{tikzcd}
\end{align*}
The second set makes sure that the inverse operation $\sigma$ behaves as expected.
\begin{align*}
    \begin{tikzcd}[ampersand replacement=\&]
        \mathbb{I}_0
        \arrow[r, "0"]
        \arrow[dr, swap, "1"]
        \&
        \mathbb{I}_1
        \arrow[d, "\sigma"]
        \&
        \mathbb{I}_0
        \arrow[l, swap, "1"]
        \arrow[dl, "0"]
        \\
        \&
        \mathbb{I}_1
    \end{tikzcd}
    \quad\quad\quad
    \begin{tikzcd}[ampersand replacement=\&]
        \mathbb{I}_0
        \arrow[equal, dr]
        \arrow[r, "\sigma"]
        \&
        \mathbb{I}_0
        \arrow[d, "\sigma"]
        \\
        \&
        \mathbb{I}_0
    \end{tikzcd}
\end{align*}
The next two axioms are coidentity and coassociativity respectively. 
\begin{align*}
    \begin{tikzcd}[ampersand replacement=\&]
        \&
        \mathbb{I}_1
        \arrow[d, "2"]
        \arrow[dl, swap, "\mathbb{I}_1"]
        \arrow[dr, "\mathbb{I}_1"]
        \&
        \\
        \mathbb{I}_1
        \&
        \mathbb{I}_2
        \arrow[r, swap, "{[1*,\mathbb{I}_1]}"]
        \arrow[l, "{[\mathbb{I}_1,0*]}"]
        \&
        \mathbb{I}_1
    \end{tikzcd}
    \quad\quad\quad
    \begin{tikzcd}[ampersand replacement=\&]
        \mathbb{I}_1
        \arrow[r, "2"]
        \arrow[d, swap, "2"] 
        \&
        \mathbb{I}_2
        \arrow[d, "{[j_1i_1,j_0 2]}"]
        \\
        \mathbb{I}_2
        \arrow[r, swap, "{[j_1 2,j_0i_0]}"]
        \&
        \mathbb{I}_1
    \end{tikzcd}
\end{align*}
Lastly, we have coinverse laws.
\begin{align*}
    \begin{tikzcd}[ampersand replacement=\&]
        \mathbb{I}_1
        \arrow[r, "2"]
        \arrow[d, swap, "*"] 
        \&
        \mathbb{I}_2
        \arrow[d, "{[\mathbb{I}_1,\sigma]}"]
        \\
        \mathbb{I}_0
        \arrow[r, swap, "1"]
        \&
        \mathbb{I}_1
    \end{tikzcd}
    \quad\quad\quad
    \begin{tikzcd}[ampersand replacement=\&]
        \mathbb{I}_1
        \arrow[r, "2"]
        \arrow[d, swap, "*"] 
        \&
        \mathbb{I}_2
        \arrow[d, "{[\sigma,\mathbb{I}_1]}"]
        \\
        \mathbb{I}_0
        \arrow[r, swap, "0"]
        \&
        \mathbb{I}_3
    \end{tikzcd}
\end{align*}

\begin{defn}\label{def:typedrealizercategory}
    A (typed) realizer category $(\mathbf{R},\mathbb{I})$ is a cartesian closed category $\mathbf{R}$ together with an interval $\mathbb{I}\in\mathbf{R}$.
\end{defn}

\begin{example}\label{eg:intervalgroupoid}
    $(\mathbf{Gpd},\mathbf{I})$ is a realizer category. The object of coarrows $\mathbf{I}_1$ is the "walking isomorphism":
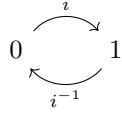
    \begin{align*}
        \begin{tikzcd}[ampersand replacement=\&]
            0
            \ar[r, bend left=50, "i"]
            \&
            1
            \ar[l, bend left=50, "i^{-1}"]
        \end{tikzcd}
    \end{align*}
    The maps $0$ and $1$ pick out the corresponding endpoints of $\mathbf{I}_1$. The map $\sigma$ sends $i\mapsto i^{-1}$.

    $\mathbf{I}_2$ has three objects and, again, one arrow in each hom set.
    
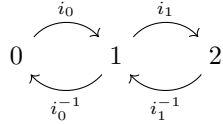
\begin{align*}
        \begin{tikzcd}[ampersand replacement=\&]
            0
            \ar[r, bend left=50, "i_0"]
            \&
            1
            \ar[l, bend left=50, "i_0^{-1}"]
            \ar[r, bend left=50, "i_1"]
            \&
            2
            \ar[l, bend left=50, "i_1^{-1}"]
        \end{tikzcd}
    \end{align*}
    The maps $i_0$ and $i_1$ send $i\in \mathbf{I}_1$ to the synonymous (eponymous, even) morphisms in $\mathbb{I}_2$. The map $2$ picks out the composite $i_1 i_0$.

    Continuing the trend, $\mathbf{I}_3$ has four objects and one arrow in each hom set.
    
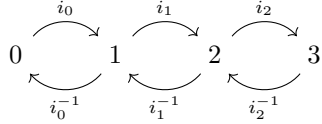
\begin{align*}
        \begin{tikzcd}[ampersand replacement=\&]
            0
            \ar[r, bend left=50, "i_0"]
            \&
            1
            \ar[l, bend left=50, "i_0^{-1}"]
            \ar[r, bend left=50, "i_1"]
            \&
            2
            \ar[l, bend left=50, "i_1^{-1}"]
            \ar[r, bend left=50, "i_2"]
            \&
            3
            \ar[l, bend left=50, "i_2^{-1}"]
        \end{tikzcd}
    \end{align*}
    The map $j_0$ sends $i_0 \mapsto i_0$ and $i_1 \mapsto i_1$; the map $j_1$ sends $i_0 \mapsto i_1$ and $i_1 \mapsto i_2$.
\end{example}

\begin{example}
    Let $\mathbf{hTop}$ be the category of spaces and homotopy classes of maps. The full subcategory spanned by the CW complexes is cartesian closed ($\mathbf{hTop}$ is only weakly cartesian closed). It contains an interval whose object of coarrows is the real unit interval $[0,1]$. Note that we must take homotopy classes of maps so that the cogroupoid axioms hold.
\end{example}

\subsubsection{Homotopies}\label{sec:homotopies}

The interval $\mathbb{I}\in\mathbf{R}$ endows the ambient category with the structure of a (2,1)-category \citep[in fact, a strict $\omega$-category, see][Theorem 1.12]{warren12}. The higher cells are given by homotopies with respect to $\mathbb{I}$.

A homotopy $H:f\Rightarrow g:A\rightarrow B$ with respect to the interval $\mathbb{I}$ is a map $H:A\times \mathbb{I}_1 \rightarrow B$ making the following diagram in $\mathbf{R}$ commute.
\begin{align*}
    \begin{tikzcd}[ampersand replacement=\&]
            A \times \mathbb{I}_0
            \arrow[d, swap, "\pi_1"]
            \arrow[r, "A\times 0"]
            \&
            A \times \mathbb{I}_1
            \arrow[d, "H"]
            \&
            A \times \mathbb{I}_0
            \arrow[d, "\pi_1"]
            \arrow[l, swap, "A\times 1"]
            \\
            A
            \arrow[r, swap, "f"]
            \&
            B
            \&
            A
            \arrow[l, "g"]
        \end{tikzcd}
\end{align*}
Given a homotopy $H:A\times \mathbb{I}_1 \rightarrow B$ we can find its domain and codomain respectively by:
\begin{align*}
    &\mathsf{dom}(H) \coloneqq H \circ \left\langle A, 0* \right\rangle
    &\mathsf{cod}(H) \coloneqq H \circ \left\langle A, 1* \right\rangle
\end{align*}

As the functor $(-)\times A$ possesses a right adjoint, the following square is a pushout.
\begin{align*}
    \begin{tikzcd}[ampersand replacement=\&]
        A\times \mathbb{I}_0
        \arrow[r, "A\times 0"]
        \arrow[d, swap, "A\times 1"] 
        \&
        A\times \mathbb{I}_1
        \arrow[d, "A\times i_1"]
        \\
        A\times \mathbb{I}_1
        \arrow[r, swap, "A\times i_0"]
        \&
        A\times \mathbb{I}_2
        \arrow[ul, phantom, "\ulcorner", very near start]
    \end{tikzcd}
\end{align*}
If $H,H': A \times \mathbb{I}_1 \rightarrow B$ such that $H'(A \times 0) = H (A \times 1)$, then the morphism $[H',H]:A\times \mathbb{I}_2 \rightarrow B$ is given by
\begin{align*}
    \mu \left[ \lambda (H' \circ \mathsf{swap}), \lambda (H \circ \mathsf{swap}) \right] \circ \mathsf{swap}
\end{align*}
If $\mathsf{dom}(H') = \mathsf{cod}(H)$ then their vertical composition is defined using this universal morphism:
\begin{equation*}
    H' \circ H \coloneqq [H',H] \circ \left(A\times 2\right)
\end{equation*}
(note the overloading of the composition symbol $\circ$).

The horizontal composition $H' \ast H$ of $H:A\times \mathbb{I}_1\rightarrow B$ and $H':B\times \mathbb{I}_1\rightarrow C$ is given by the following composite in $\mathbf{R}$.
\begin{align*}
    \begin{tikzcd}[ampersand replacement=\&]
        A\times \mathbb{I}_1
        \arrow[r, "A\times\Delta"]
        \&
        A\times(\mathbb{I}_1 \times \mathbb{I}_1)\cong(A\times \mathbb{I}_1)\times \mathbb{I}_1
        \arrow[r, "H\times \mathbb{I}_1"]
        \&
        B\times \mathbb{I}_1
        \arrow[r, "H'"]
        \&
        C
    \end{tikzcd}
\end{align*}

The identity homotopy at $f:A\rightarrow B$ is given by $f\pi_1: A\times \mathbb{I}_1 \rightarrow B$, and the inverse of a homotopy $H:A\times \mathbb{I}_1 \rightarrow B$ is $H\circ (A\times\sigma)$.

\subsubsection{Fundamental groupoids}\label{sec:fundamentalgroupoids}

By considering maps out of $\mathbb{I}$ into a fixed object $A\in\mathbf{R}$ we obtain the fundamental groupoid of $A$. That is, we have a 2-functor
\begin{equation*}
    \Pi = (-)^\mathbb{I} : \mathbf{R} \rightarrow \mathbf{Gpd}
\end{equation*}
A quick way to see this is because the contravariant hom functor takes colimits (used in the definition of an interval) to limits (used in the definition of a category).

The fundamental groupoid of $A$ has as objects points in $A$, ie. maps $\mathbb{I}_0\rightarrow A$. A morphism $\alpha:a\rightarrow b$ is a path $\alpha$ in $A$ making the following diagram commute.
\begin{align*}
    \begin{tikzcd}[ampersand replacement=\&]
    \mathbb{I}_0
    \arrow[r, "0"]
    \arrow[dr, swap, "a"]
    \&
    \mathbb{I}_1
    \arrow[d, "\alpha"]
    \&
    \mathbb{I}_0
    \arrow[l, swap, "1"]
    \arrow[dl, "b"]
    \\
    \&
    A
    \end{tikzcd}
\end{align*}
The composition of $\alpha:a\rightarrow b$ with $\beta:b\rightarrow c$ is defined by
    \begin{equation*}
        \beta \circ \alpha \coloneqq [\beta,\alpha] \circ 2
    \end{equation*}
($2$ re-parameterises the double-length path). The identity at $a$ is $a*$ and the inverse of $\alpha$ is $\alpha\sigma$.

If $f:A\rightarrow B \in \mathbf{R}$, the functor $\Pi(f):\Pi A \rightarrow \Pi B$ is given by post-composition (in $\mathbf{R}$) with $f$. The composition law for functors holds because $f[\beta,\alpha] = [f\beta,f\alpha]$ by the universal property of the pushout (\ref{eqn:pushout}).

A natural isomorphism $\phi: F\Rightarrow G: \mathbf{C}\rightarrow \mathbf{D}$ between functors (between categories or groupoids) is equivalently given by a functor:
\begin{align*}
    \phi:\mathbf{C} \times \mathbf{I}_1 \rightarrow \mathbf{D}
\end{align*}
such that $\phi(-,0)=F$ and $\phi(-,1)=G$. In this way, a homotopy $H:f\Rightarrow g:A\rightarrow B\in\mathbf{R}$ gives rise to a natural transformation
    \begin{equation*}
        \Pi(H) : \Pi A \times \mathbf{I}_1 \rightarrow \Pi B
    \end{equation*}
by setting:
\begin{align*}
    &\Pi(H)(a,0) = \Pi(H)(a) \coloneqq H \circ \langle a,0 \rangle : \mathbb{I}_0\rightarrow B
    \\
    &\Pi(H)(a,1) = \Pi(H)(a) \coloneqq H \circ \langle a,1 \rangle : \mathbb{I}_0\rightarrow B
    \\
    &\Pi(H)(\alpha,i) = \Pi(H)(\alpha) \coloneqq H \circ \langle \alpha,\mathbb{I}_1 \rangle : \mathbb{I}_1\rightarrow B
    \\
    &\Pi(H)(\alpha,i^{-1}) = \Pi(H)(\alpha) \coloneqq H \circ \langle \alpha,\sigma \rangle : \mathbb{I}_1\rightarrow B
\end{align*}

Given a homotopy $H:f\Rightarrow g:A\rightarrow B$ and a path $\alpha:a\rightarrow b\in \Pi A$, the path $\Pi(H)(\alpha,i) \in \Pi B$ can be thought of as the diagonal path through a naturality square. Indeed, we have the following.
\begin{lem}[Warren, personal communication]\label{thm:boundaries}
With the above data, the following diagram commutes in $\Pi B$.
\begin{align*}
    \begin{tikzcd}[ampersand replacement=\&]
    \Pi(f)(a)
    \arrow[rr, "{H \circ \langle \alpha 0 * , \mathbb{I}_1 \rangle}"]
    \arrow[dd, swap, "{H\circ \langle \alpha, 0* \rangle}"]
    \arrow[rrdd, "{\Pi(H)(\alpha,i)}"]
    \&
    \&
    \Pi(g)(a)
    \arrow[dd, "{H\circ \langle \alpha, 1* \rangle}"]
    \\
    \\
    \Pi(f)(b)
    \arrow[rr, swap, "{H \circ \langle \alpha 1 *, \mathbb{I}_1 \rangle}"]
    \&
    \&
    \Pi(g)(b)
    \end{tikzcd}
\end{align*}
\end{lem}
\begin{proof}
We show that the diagonal is equal to the left-bottom boundary; that these are equal to the top-right boundary is a symmetric argument. Using the definition of composition in $\Pi B$, we have to show
\begin{align*}
    H \circ \langle \alpha, \mathbb{I}_1 \rangle = \left[H \circ \langle \alpha 1 *, \mathbb{I}_1 \rangle, H \circ \langle \alpha, 0 *\rangle\right] \circ 2
\end{align*}
This simplifies to
\begin{align*}
    H \circ (\alpha \times \mathbb{I}_1) \circ \left[ \left\langle 1 *, \mathbb{I}_1 \rangle, \langle \mathbb{I}_1, 0 * \right\rangle \right] \circ 2
\end{align*}
But
\begin{align*}
    \left[\left\langle 1 *, \mathbb{I}_1 \rangle, \langle \mathbb{I}_1, 0 * \right\rangle\right] \circ 2
\end{align*}
is a decomposition of the diagonal $\Delta_{\mathbb{I}_1}$ \cite[see][p. 212]{warren12}.
\end{proof}

\subsubsection{Squares}\label{sec:squares}

Let $(\mathbf{R},\mathbb{I})$ be a realizer category. Define the (edge-symmetric) double category $\Box\mathbf{R}(A,B)$ to have:
\begin{itemize}
    \item objects: maps $A\rightarrow B \in \mathbf{R}$;
    \item both horizontal and vertical morphisms: maps $A\times \mathbb{I}_1 \rightarrow B$;
    \item 2-cells: commutative squares of maps $A\times \mathbb{I}_1 \rightarrow B$.
\end{itemize}
Denote by $\blacksquare\mathbf{R}(A,B)$ the (edge-symmetric) double category with the same objects and morphisms as $\Box\mathbf{R}(A,B)$, but with 2-cells maps $A\times \mathbb{I}_1 \times \mathbb{I}_1 \rightarrow B$.

As a consequence of \Cref{thm:boundaries}, every 2-cell $\phi\in \blacksquare\mathbf{R}(A,B)$ determines a 2-cell $\partial(\phi)\in \Box\mathbf{R}(A,B)$:
\begin{align*}
    \begin{tikzcd}[ampersand replacement=\&]
        \phi_{00}
        \arrow[r, "\phi_{\mathbb{I}_1 0}"{name=A}]
        \arrow{d}[swap]{\phi_{0 \mathbb{I}_1}}
        \&
        \phi_{10}
        \arrow{d}{\phi_{1 \mathbb{I}_1}}
        \\
        \phi_{01}
        \arrow[r, swap, "\phi_{\mathbb{I}_1 1}"{name=B}]
        \&
        \phi_{11}
        \arrow[from=A,to=B, Rightarrow, shorten >=2ex, shorten <=2ex, "\partial(\phi)"]
    \end{tikzcd}
\end{align*}
where, eg. $\phi_{01} \coloneqq \phi \circ \langle A, 0*, 1* \rangle: A \rightarrow B$ and $\phi_{\mathbb{I}_1 1} \coloneqq \phi \circ \langle A\pi_1, \mathbb{I}_1\pi_2, 1*\pi_2 \rangle: A\times \mathbb{I}_1 \rightarrow B$. This determines a double functor:
\begin{align*}
    \partial_{\mathbf{R}(A,B)}: \blacksquare\mathbf{R}(A,B) \rightarrow \Box\mathbf{R}(A,B)
\end{align*}
for each $A,B\in\mathbf{R}$ (we sometimes just write $\partial$).
%We will write $\partial_\mathbf{R}$---or even just $\partial$---to refer either to a particular $\partial_{\mathbf{R}(A,B)}$ or to the whole family of these double functors.

\begin{rem}\label{rem:squares}
    Some of the results to follow rely on the hypothesis that, for all $A,B\in\mathbf{R}$, $\partial_{\mathbf{R}(A,B)}$ is an isomorphism of double categories. As shorthand for this we will say that $\partial_\mathbf{R}$ is an isomorphism.
    
    This hypothesis holds whenever $(\mathbf{R},\mathbb{I})$ is finitely complete as a (2,1)-category, for then the cotensor $\mathbf{I}_1 \pitchfork B$ is necessarily isomorphic to the internal hom $B^{\mathbb{I}_1}$ \citep[][Lemma 3.25]{warren08}:
    \begin{align*}
        \mathbf{R}\left(A, \mathbf{I}_1 \pitchfork B \right) \cong \mathbf{Gpd}\left( \mathbf{I}_1, \mathbf{R}(A,B) \right) \cong \mathbf{R}\left(A\times\mathbb{I}_1, B \right) \cong \mathbf{R}\left(A,B^{\mathbb{I}_1}\right)
    \end{align*}
\end{rem}

\subsubsection{Untyped realizer categories}\label{sec:untypedrealizercategories}

We know from other settings \citep{birkedal00a,birkedal00b,robinson01,lietz02} that impredicativity is intimately related to untypedness at the level of realizers. A "universal object" allows us to turn the \textit{a priori} typed notion of realizability given by a category into an untyped one. Traditionally, an object in a category is universal iff every object in the category is a retract of it. In the higher setting, we allow for pseudoretracts.
\begin{defn}\label{def:universalobject}
    An object $U\in\mathbf{R}$ is universal iff for every object $A\in\mathbf{R}$ there is a section and retraction
    \begin{align*}
        \begin{tikzcd}[ampersand replacement=\&]
            A \ar[r, "s_A"] \& U \ar[r, "r_A"] \& A
        \end{tikzcd}
    \end{align*}
    together with a homotopy
    \begin{align*}
        \rho_A: r_A s_A \Rightarrow \mathsf{id}_A
    \end{align*}
\end{defn}

\begin{defn}\label{def:untypedrealizercategory}
    An untyped realizer category $(\mathbf{R},\mathbb{I},U)$ is a (typed) realizer category $(\mathbf{R},\mathbb{I})$ together with a universal object $U\in\mathbf{R}$.
\end{defn}
In particular, an untyped realizer category provides an up-to-homotopy model of the untyped $\lambda$-calculus: $U^U$ is a pseudoretract of $U$.

\begin{example}
    Domain theory is a rich source of models of the untyped $\lambda$-calculus. \textit{Scott complete categories}, introduced by \cite{adamek97}, are categorified \textit{Scott domains} \citep{scott93}. A Scott complete category is a finitely accessible category in which every diagram with a cocone has a colimit. Let $\kappa$ be an inaccessible cardinal. The category $\mathbf{SCC}$ of locally $\kappa$-small Scott complete categories and continuous functors is cartesian closed \citep{adamek97}, contains the interval $\mathbf{I}\in\mathbf{Gpd}$ (see Example \ref{eg:intervalgroupoid}) and contains a universal object \citep{velebil99}. Homotopies with respect to $\mathbf{I}$ in $\mathbf{SCC}$ are arbitrary natural isomorphisms. $\partial_{\mathbf{Gpd}}$ is an isomorphism and thus $\partial_{\mathbf{SCC}}$ is too.
\end{example}

\section{Categories of partitioned groupoidal assemblies}\label{sec:categoriesofpartitionedgroupoidalassemblies}

We first define the constituents of the category $\mathbf{PGAsm}(\mathbf{R},\mathbb{I})$ of partitioned groupoidal assemblies over the realizer category $(\mathbf{R},\mathbb{I})$.
\begin{defn}
\label{def:partitionedgroupoidalassemblies}
    A partitioned groupoidal assembly $X$ is a triple $(X,A,{\left\Vert - \right\Vert}_X)$, where $X\in\mathbf{Gpd}$ is the underlying groupoid, $A\in\mathbf{R}$ is the "realizer type", and
    \begin{align*}
        {\left\Vert - \right\Vert}_X : X \rightarrow \Pi A
    \end{align*}
    is the functor that assigns realizers to objects and isomorphisms of $X$. A partitioned groupoidal assembly $X$ is modest whenever ${\left\Vert-\right\Vert}_X$ is fully faithful.
    
    A morphism
    \begin{align*}
        F: (X,A,{\left\Vert-\right\Vert}_X) \rightarrow (Y,B,{\left\Vert-\right\Vert}_Y)
    \end{align*}
    of partitioned groupoidal assemblies is a functor $F: X \rightarrow Y$ such that there exists $e:A\rightarrow B$ and a natural isomorphism $\epsilon$ as in the following diagram.
    \begin{align}\label{eqn:morphismpgasm}
        \begin{tikzcd}[ampersand replacement=\&]
            X
            \arrow[rr, "F"]
            \arrow[dd, swap, "{\left\Vert-\right\Vert}_X"]
            \& \&
            Y
            \arrow[dd, "{\left\Vert-\right\Vert}_Y"]
            \ar[ddll, Leftarrow, shorten <= 1.5em, shorten >= 1.5em, "\epsilon"]
            \\ \\
            \Pi A
            \arrow[rr, swap, "\Pi(e)"]
            \& \&
            \Pi B
        \end{tikzcd}
    \end{align}
\end{defn}

This formalizes the homotopy BHK interpretation: a realizer of an object $x\in X$ is a point $a\in\Pi A$ (ie. a map $a:\mathbb{I}_0 \rightarrow A$), and a realizer of an identification (isomorphism) $p:x\rightarrow x' \in X$ is a path $\alpha: a \rightarrow a' \in \Pi A$ (ie. a map $\alpha:\mathbb{I}_1\rightarrow A$ such that $\alpha0=a$ and $\alpha1=a'$, where $a'$ realizes $x'$.

The morphism $e: A \rightarrow B$ implements the functor $F$ up to natural isomorphism, as witnessed by $\epsilon$. We will often refer to the \textit{pair} $(e,\epsilon)$ as a realizer for $F$, writing $(e,\epsilon) \Vdash F$.

Identities in $\mathbf{PGAsm}(\mathbf{R},\mathbb{I})$ are identity functors; each is realized by the relevant pair of an identity morphism and an identity natural isomorphism. Composition is performed by pasting squares such as (\ref{eqn:morphismpgasm}) side by side ; explicitly, if $(e,\epsilon)\Vdash F:X\rightarrow Y$ and $(e',\epsilon')\Vdash G:Y\rightarrow Z$ then the composite $GF:X\rightarrow Z$ is realized by
\begin{align*}
    \left( e'e, \left(\epsilon' \ast F \right) \circ \left(\Pi\left(e'\right) \ast \epsilon\right) \right)
\end{align*}

Observe that $\mathbf{PGAsm}(\mathbf{R},\mathbb{I})$ is a quotient of the "super-comma" category $\mathbf{Gpd} \downarrow \Pi$ (to use the terminology of \cite{maclanecwm}---though might nowadays be called the underlying 1-category of the 2-comma category); it is a quotient because realizers of are only required to exist, not carried around as extra data. This should be compared to the $\mathcal{F}$-construction of \cite{robinson01}, which is a set-based, 1-categorical analogue.

The above definition of modesty generalizes the traditional notion, for if $\mathbb{I}$ is a discrete interval and $X$ a set then, given fullness, ${\left\Vert x \right\Vert}_X = {\left\Vert x' \right\Vert}_X$ implies $X(x,x')\neq\emptyset$ and thus  $x=x'$.

\begin{prop}\label{thm:pgasmcartesianclosure}
    $\mathbf{PGAsm}(\mathbf{R},\mathbb{I})$ is weakly cartesian closed. Moreover, the weak exponential object $Y^X$ is modest whenever $Y$ is.
\end{prop}

\begin{proof}
    The terminal object is $(\mathbf{1},\mathbb{I}_0,\Vert \!-\! \Vert_\mathbf{1}: *\mapsto \mathbb{I}_0)$. Let $(X,A,{\left\Vert-\right\Vert}_X), (Y,B,{\left\Vert-\right\Vert}_Y) \in \mathbf{PGAsm}(\mathbf{R},\mathbb{I})$. Their product is
    \begin{align*}
        \left( X\times Y, A\times B, {\left\Vert-\right\Vert}_{X\times Y} \right)
    \end{align*}
    where
    \begin{align*}
        {\left\Vert-\right\Vert}_{X\times Y} \coloneqq \left\langle {\left\Vert \pi_1(-) \right\Vert}_X , {\left\Vert \pi_2(-) \right\Vert}_Y \right\rangle
    \end{align*}
    
    The weak exponential is
    \begin{align*}
        Y^X \coloneqq \left( \mathsf{Real}\left(Y^X\right), B^A, {\left\Vert-\right\Vert}_{Y^X} \right)
    \end{align*}
    where $\mathsf{Real}(Y^X)$ is the groupoid whose objects are triples
    \begin{align*}
        \left( F: X \rightarrow Y, e \in \Pi\left(B^A\right),\epsilon \right)
    \end{align*}
    such that
    \begin{align*}
        \left( \mu(e) \circ \langle *, A \rangle, \epsilon \right) \Vdash F
    \end{align*}
    and whose morphisms $(F,e,\epsilon)\rightarrow(G,e',\epsilon')$ are tuples
    \begin{align*}
        \left( \psi:F\Rightarrow G, f:e\rightarrow e', \zeta \right)
    \end{align*}
    such that
    \begin{align*}
        \left( \mu(f)\circ\mathsf{swap}, \zeta \right) \Vdash \psi
    \end{align*}
    and:
    \begin{align}\label{eqn:real2}
        &\zeta(-,0,-) = \epsilon
        &\zeta(-,1,-) = \epsilon'
    \end{align}
    Actually, the natural isomorphism $\zeta$ is uniquely determined: the equations (\ref{eqn:real2}) mean that the boundary of $\zeta$ is fully specified and $\partial_\mathbf{Gpd}$ is an isomorphism.

    The evaluation morphism is defined:
    \begin{align*}
        &\mathsf{ev}: Y^X \times X \rightarrow Y
        \\
        &\mathsf{ev}(F,e,\epsilon,x) \coloneqq Fx
        \\
        &\mathsf{ev}(\psi,f,p) \coloneqq \psi(p,i)
    \end{align*}
    and realized by $(\mathsf{ev},\epsilon')$, where $\mathsf{ev}$ denotes the evaluation map in $\mathbf{R}$, and where:
    \begin{align*}
        \epsilon'_{(F,e,\epsilon,x)} \coloneqq \epsilon_x
    \end{align*}
    The naturality of $\epsilon'$ is captured the following cube.
    \begin{align*}
        \begin{tikzcd}[ampersand replacement=\&]
        \& \&
        \substack{\mathsf{ev} \circ \langle e, {\left\Vert x \right\Vert}_X \rangle \\ = \Pi(\mu(e)) {\left\Vert x \right\Vert}_X}
        \arrow[rrr, "\epsilon_x"]
        \arrow{ddll}[swap]{\substack{\mathsf{ev} \circ \langle e, {\left\Vert p \right\Vert}_X \rangle \\ = \Pi(\mu(e)) {\left\Vert p \right\Vert}_X}}
        \ar{dddd}[near start]{\substack{\mathsf{ev} \circ \langle f, {\left\Vert x \right\Vert}_X \rangle \\ = \Pi(\mu(f)) {\left\Vert x \right\Vert}_X}}
        \& \& \&
        {\left\Vert Fx \right\Vert}_Y
        \ar{dddd}{{\left\Vert \psi_x \right\Vert}_Y}
        \ar{ddll}{{\left\Vert F(p) \right\Vert}_Y}
        \\ \\
        \substack{\mathsf{ev} \circ \langle e, {\left\Vert x' \right\Vert}_X \rangle \\ = \Pi(\mu(e)) {\left\Vert x' \right\Vert}_X}
        \ar[crossing over]{rrr}[near start, swap]{\epsilon_{x'}}
        \ar{dddd}[swap]{\substack{\mathsf{ev} \circ \langle f, {\left\Vert x' \right\Vert}_X \rangle \\ = \Pi(\mu(f)) {\left\Vert x' \right\Vert}_X}}
        \& \& \&
        {\left\Vert Fx' \right\Vert}_Y
        \\ \\
        \& \&
        \substack{\mathsf{ev} \circ \langle e', {\left\Vert x \right\Vert}_X \rangle \\ = \Pi(\mu(e')) {\left\Vert x \right\Vert}_X}
        \arrow{rrr}[near start]{\epsilon'_{x}}
        \arrow{ddll}[near start, swap]{\substack{\mathsf{ev} \circ \langle e', {\left\Vert p \right\Vert}_X \rangle \\ = \Pi(\mu(e')) {\left\Vert p \right\Vert}_X}}
        \& \& \&
        {\left\Vert Gx \right\Vert}_Y
        \ar{ddll}{{\left\Vert G(p) \right\Vert}_Y}
        \\ \\
        \substack{\mathsf{ev} \circ \langle e', {\left\Vert x' \right\Vert}_X \rangle \\ = \Pi(\mu(e')) {\left\Vert x' \right\Vert}_X}
        \arrow[rrr, swap, "\epsilon'_{x'}"]
        \& \& \&
        \Vert Gx' \Vert_Y
        \ar[from=uuuu,crossing over]{}[near start]{{{\left\Vert \psi_{x'} \right\Vert}_Y}}
    \end{tikzcd}
    \end{align*}

    For the universal property, suppose we have a morphism
    \begin{align*}
        K: \left( Z, C, {\left\Vert-\right\Vert}_Z \right) \times X \rightarrow Y
    \end{align*}
    Pick a realizer $(e,\epsilon)\Vdash K$. With this we define the transpose $\widetilde{K}: Z \rightarrow Y^X$ of $K$. On objects:
    \begin{align*}
        \widetilde{K}(z) \coloneqq \left( K(z,-), e^z, \epsilon^z \right)
    \end{align*}
    where:
    \begin{align*}
        &e^z \coloneqq \lambda\left( e \circ \left( {\left\Vert z \right\Vert}_Z \times A \right) \right)
        \\
        &\epsilon^z_x \coloneqq {\epsilon}_{(z,x)}
    \end{align*}
    This works because
    \begin{align*}
        \Pi \left( \mu\left( \lambda \left( e \circ \left( {\left\Vert z \right\Vert}_Z \times A \right) \right) \right) \circ \langle *,A \rangle \right) {\left\Vert x \right\Vert}_X
        = &\mu\left( \lambda \left( e \circ \left( {\left\Vert z \right\Vert}_Z \times A \right) \right) \right) \circ \langle *,A \rangle \circ {\left\Vert x \right\Vert}_X
        \\
        = & e \circ \left( {\left\Vert z \right\Vert}_Z \times A \right) \langle *, A \rangle \circ {\left\Vert x \right\Vert}_X
        \\
        = & e \circ \left( {\left\Vert z \right\Vert}_Z \times A \right) \circ \langle \mathbb{I}_0, {\left\Vert x \right\Vert}_X \rangle
        \\
        = & e \circ \left\langle {\left\Vert z \right\Vert}_Z, {\left\Vert x \right\Vert}_X \right\rangle
        \\
        = & \Pi(e) {\left\Vert (z,x) \right\Vert}_{Z\times X}
    \end{align*}
    and
    \begin{align*}
        \epsilon_{(z,x)}: \Pi(e) {\left\Vert (z,x) \right\Vert}_{Z\times X} \rightarrow {\left\Vert K(z,x) \right\Vert}_Y
    \end{align*}
    On morphisms:
    \begin{align*}
        \widetilde{K}(r) \coloneqq \left( \psi^r, f^r \right)
    \end{align*}
    where:
    \begin{align*}
        \psi^r(p,i) &\coloneqq K(r,p)
        \\
        f^r &\coloneqq \lambda\left( e \circ \left( {\left\Vert r \right\Vert}_Z \times A \right) \right)
    \end{align*}

    Suppose that $Y$ is modest and let $(F,e,\epsilon),(G,e',\epsilon') \in \mathsf{Real}(Y^X)$. Any $f:e\rightarrow e'$ uniquely determines a morphism $(\psi,f): (F,e,\epsilon) \rightarrow (G,e',\epsilon')$ due to ${\left\Vert-\right\Vert}_Y$ being fully faithful: the image of the component $\psi_x$ under ${\left\Vert-\right\Vert}_Y$ is the unique morphism making the following square commute.
    \begin{align*}
        \begin{tikzcd}[ampersand replacement=\&]
            {\left\Vert Fx \right\Vert}_Y
            \arrow[rr, dashed, "{\left\Vert \psi_x \right\Vert}_Y"]
            \& \&
            {\left\Vert Gx \right\Vert}_Y
            \\ \\
            \substack{\Pi(f) \langle 0, {\left\Vert x \right\Vert}_X \rangle \\ = \Pi(e){\left\Vert x \right\Vert}_X}
            \arrow[uu, "\epsilon_x"]
            \arrow[rr, swap,  "{\Pi(f) \left\langle \mathbb{I}_1, {\left\Vert x \right\Vert}_X \right\rangle}"]
            \& \&
            \substack{\Pi(f) \langle 1, {\left\Vert x \right\Vert}_X \rangle \\ = \Pi(e'){\left\Vert x \right\Vert}_X}
            \arrow[uu, swap, "\epsilon'_x"]
        \end{tikzcd}
    \end{align*}
\end{proof}

\subsection{As 2-categories}\label{sec:asa2category}

$\mathbf{PGAsm}(\mathbf{R},\mathbb{I})$ possesses an interval of its own. The object of coarrows is given by:
\begin{align*}
    &\mathbf{I}_1 \coloneqq \left( \mathbf{I}_1, \mathbb{I}_1, {\left\Vert-\right\Vert}_{\mathbf{I}_1} \right)
    \\
    &{\left\Vert 0 \right\Vert}_{\mathbf{I}_1} \coloneqq 0
    &&{\left\Vert 1 \right\Vert}_{\mathbf{I}_1} \coloneqq 1
    &&&{\left\Vert i \right\Vert}_{\mathbf{I}_1} \coloneqq \mathbb{I}_1
\end{align*}
The other parts of the cogroupoid diagram are likewise obtained by marrying up the corresponding parts of the cogroupoids in $\mathbf{Gpd}$ and $\mathbf{R}$.
\begin{align*}
    &\mathbf{I}_2 \coloneqq \left( \mathbf{I}_2, \mathbb{I}_2, {\left\Vert-\right\Vert}_{\mathbf{I}_2} \right)
    \\
    &{\left\Vert 0 \right\Vert}_{\mathbf{I}_2} \coloneqq i_0 0
    &&{\left\Vert 1 \right\Vert}_{\mathbf{I}_2} \coloneqq i_0 1 = i_1 0
    &&&{\left\Vert 2 \right\Vert}_{\mathbf{I}_2} \coloneqq i_1 1
    \\
    &{\left\Vert i_0 \right\Vert}_{\mathbf{I}_2} \coloneqq i_0
    \quad\quad
    &&{\left\Vert i_1 \right\Vert}_{\mathbf{I}_2} \coloneqq i_1
    \\
    &\mathbf{I}_3 \coloneqq \left( \mathbf{I}_3, \mathbb{I}_3, {\left\Vert-\right\Vert}_{\mathbf{I}_3} \right)
    \\
    &{\left\Vert 0 \right\Vert}_{\mathbf{I}_2} \coloneqq j_0 i_0 0
    &&{\left\Vert 1 \right\Vert}_{\mathbf{I}_2} \coloneqq j_0 i_0 1 = j_0 i_1 0 = j_1 i_0 0
    \\
    &{\left\Vert 2 \right\Vert}_{\mathbf{I}_2} \coloneqq j_0 i_1 1 = j_1 i_0 1 = j_1 i_1 0
    &&{\left\Vert 3 \right\Vert}_{\mathbf{I}_2} \coloneqq j_1 i_1 1
    \\
    &{\left\Vert i_0 \right\Vert}_{\mathbf{I}_2} \coloneqq j_0 i_0
    &&{\left\Vert i_1 \right\Vert}_{\mathbf{I}_2} \coloneqq j_0 i_1 = j_1 i_0
    &&&{\left\Vert i_2 \right\Vert}_{\mathbf{I}_2} \coloneqq j_1 i_1
\end{align*}
The underlying functors of the morphisms $i_0,i_1,2,j_0,j_1 \in \mathbf{PGAsm}(\mathbf{R},\mathbb{I})$ are as in Example \ref{eg:intervalgroupoid} and are realized by the maps $i_0,i_1,2,j_0,j_1\in\mathbf{R}$ respectively.

We will show that $\mathbf{I}_2$ is the pushout of $0,1$; a similar argument works for $\mathbf{I}_3$. Suppose that we are in the following situation:
\begin{align*}
    \begin{tikzcd}[ampersand replacement=\&]
        \mathbf{1}
        \arrow[r, "0"]
        \arrow[d, swap, "1"] 
        \&
        \mathbf{I}_1
        \arrow[d, "i_1"]
        \arrow[ddr, bend left, "G"]
        \\
        \mathbf{I}_1
        \arrow[r, swap, "i_0"]
        \arrow[drr, bend right, swap, "F"]
        \&
        \mathbf{I}_2
        %\arrow[ul, phantom, "\ulcorner", very near start]
        %\arrow[dr, swap, "{[f\beta,f\alpha]}"]
        \\
        \& \&
        \left( X , A , {\left\Vert-\right\Vert}_X \right)
    \end{tikzcd}
\end{align*}
where $(e,\epsilon)\Vdash F$. The functor $[G,F]:\mathbf{I}_2 \rightarrow X$ is the universal morphism in $\mathbf{Gpd}$. Define $d \coloneqq e0*$ and observe that
\begin{align*}
    \Pi(d) {\left\Vert 0 \right\Vert}_{\mathbf{I}_2} = \Pi(d) {\left\Vert 1 \right\Vert}_{\mathbf{I}_2} = \Pi(d) {\left\Vert 2 \right\Vert}_{\mathbf{I}_2} = e0
\end{align*}
as well as
\begin{align*}
    \Pi(d) {\left\Vert i_0 \right\Vert}_{\mathbf{I}_2} = \Pi(d) {\left\Vert i_1 \right\Vert}_{\mathbf{I}_2} = e0* = \mathsf{id}_{e0} 
\end{align*}
We can realize $[G,F]$ with $(d,\delta)$, where:
\begin{align*}
    \delta_0 &\coloneqq \epsilon_0 : e0 \rightarrow {\left\Vert F0 \right\Vert}_X
    \\
    \delta_1 &\coloneqq \epsilon_1 e : e0 \rightarrow {\left\Vert F1 \right\Vert}_X
    \\
    \delta_2 &\coloneqq  {\left\Vert G(i) \right\Vert}_X \epsilon_1 e : e0 \rightarrow {\left\Vert G1 \right\Vert}_X
\end{align*}
Naturality of $\epsilon$ guarantees naturality of $\delta$.

As $\mathbf{PGAsm}(\mathbf{R},\mathbb{I})$ is only \textit{weakly} cartesian closed, we cannot apply \citep[Theorem 1.12]{warren12} (which assumes cartesian closedness) to deduce that it becomes a (2,1)-category when 2-cells are taken to be homotopies. However, we will show by hand that this is the case after all.

Let $\phi: F \Rightarrow G: X \rightarrow Y$ and $\psi: G \Rightarrow H: X \rightarrow Y$ be homotopies in $\mathbf{PGAsm}(\mathbf{R},\mathbb{I})$. Their vertical composition $\psi\circ \phi = \psi\phi: F \Rightarrow H: X \rightarrow Y$ is defined as in $\mathbf{Gpd}$:
\begin{align*}
    &\psi\phi : X \times \mathbf{I}_1 \rightarrow Y
    \\
    &\psi\phi(p:x\rightarrow x',i) \coloneqq \psi_{x'} \circ \phi_{x'} \circ F(p) = H(p) \circ \psi_x \circ \phi_x
\end{align*}

Let $(e^\phi, \epsilon^\phi) \Vdash \phi$ and $(e^\psi, \epsilon^\psi) \Vdash \psi$. We obtain a realizer $(e^{\psi\phi}, \epsilon^{\psi\phi}) \Vdash \psi\phi$ as follows. First let $e^{\psi\phi} \coloneqq e^\phi$. Then define:
\begin{align*}
    \epsilon^{\psi\phi}_{(x,0)} &\coloneqq \epsilon^\phi_{(x,0)} : e^\phi {\left\Vert (x,0) \right\Vert}_{X \times \mathbf{I}_1} \rightarrow {\left\Vert Fx \right\Vert}_Y
    \\
    \epsilon^{\psi\phi}_{(x,1)} &\coloneqq {\left\Vert \psi_x \right\Vert}_Y \circ \epsilon^\phi_{(x,1)} : e^\phi {\left\Vert (x,0) \right\Vert}_{X \times \mathbf{I}_1} \rightarrow {\left\Vert Gx \right\Vert}_Y \rightarrow {\left\Vert Fx \right\Vert}_Y
\end{align*}
That this is natural is captured in the following diagram.
\begin{align*}
        \begin{tikzcd}[ampersand replacement=\&]
            e^\phi {\left\Vert (x,0) \right\Vert}_{X \times \mathbf{I}_1}
            \arrow[ddd, "{e^\phi {\left\Vert (p,0) \right\Vert}_{X \times \mathbf{I}_1}}"]
            \arrow[rrr, "{\epsilon^\psi\phi_{(x,0)} = \epsilon^\phi_{(x,0)}}"]
            \arrow[dddddd, bend right=60, swap, "{e^\phi {\left\Vert (p,i) \right\Vert}_{X \times \mathbf{I}_1}}"]
            \& \& \&
            {\left\Vert Fx \right\Vert}_Y
            \arrow[ddd, swap, "\substack{{\left\Vert \psi\phi(p,0) \right\Vert}_Y \\ = {\left\Vert F(p) \right\Vert}_Y}"]
            \arrow[ddddddddd, bend left=40, "{{\left\Vert \psi\phi(p,i) \right\Vert}_Y}"]
            \\ \\ \\
            e^\phi {\left\Vert (x',0) \right\Vert}_{X \times \mathbf{I}_1}
            \arrow[rrr, "{\epsilon^\phi_{(x',0)}}"]
            \arrow[ddd, "{e^\phi {\left\Vert (x',i) \right\Vert}_{X \times \mathbf{I}_1}}"]
            \& \& \&
            {\left\Vert Fx' \right\Vert}_Y
            \arrow[ddd, swap, "\substack{{\left\Vert \psi\phi(p,0) \right\Vert}_Y \\ = {\left\Vert F(p) \right\Vert}_Y}"]
            \\ \\ \\
            e^\phi {\left\Vert (x',1) \right\Vert}_{X \times \mathbf{I}_1}
            \arrow[rrr, swap, "\epsilon^\phi_{(x',1)}"]
            \arrow[dddrrr, swap, "\epsilon^{\psi\phi}_{(x',1)}"]
            \& \& \&
            {\left\Vert Gx' \right\Vert}_Y
            \arrow[ddd, swap, "{\left\Vert \psi_{x'} \right\Vert}_Y"]
            \\ \\ \\
            \& \& \&
            {\left\Vert Hx' \right\Vert}_Y
        \end{tikzcd}
    \end{align*}

Now let $\phi: F\Rightarrow G: X\rightarrow Y$ and $\psi: H\Rightarrow K: Y\rightarrow Z$. Their horizontal composition $\psi \ast \phi: HF \Rightarrow KG: X \rightarrow Z$ is defined as in $\mathbf{Gpd}$:
\begin{align*}
    &\psi\ast\phi : X \times \mathbf{I}_1 \rightarrow Z
    \\
    &\psi\ast\phi(x,0) \coloneqq HFx
    \\
    &\psi\ast\phi(x,1) \coloneqq KGx
    \\
    &\psi\ast\phi(p,i) \coloneqq \psi_{Gx'} \circ H\left(\phi_{x'}\right) \circ H(F(p)) = K(G(p)) \circ \psi_{Gx} \circ H\left(\phi_x\right)
\end{align*}

Let $(e^\phi, \epsilon^\phi) \Vdash \phi$, $(e^\psi, \epsilon^\psi) \Vdash \psi$ and $(e^H, \epsilon^H) \Vdash H$. We obtain a realizer $(e^{\psi\ast\phi}, \epsilon^{\psi\ast\phi}) \Vdash \psi\ast\phi$ as follows.
\begin{align*}
    e^{\psi\ast\phi} &\coloneqq e^H e^\phi
    \\
    \epsilon^{\psi\ast\phi}_{(x,0)} &\coloneqq \left( \left(\epsilon^H \ast \phi \right) \circ \left(\Pi\left(e^H\right) \ast \epsilon^\phi \right) \right)_{(x,0)}
    \\
    \epsilon^{\psi\ast\phi}_{(x,1)} &\coloneqq \left\Vert \psi_{Gx} \right\Vert_Z \circ \left( \left(\epsilon^H \ast \phi \right) \circ \left(\Pi\left(e^H\right) \ast \epsilon^\phi \right) \right)_{(x,1)}
\end{align*}
Naturality is captured in the following diagram.
\begin{align*}
    \adjustbox{scale=0.95}{
    \begin{tikzcd}[ampersand replacement=\&]
        \Pi \left( e^H \right) \left( e^\phi \Vert (x,0) \Vert_{X \times \mathbf{I}_1} \right)
        \arrow{rr}{{\Pi \left( e^H \right) \left( \epsilon^\phi_{(x,0)} \right)}}
        \arrow{dd}{{\Pi \left( e^H \right) \left( e^\phi \Vert (p,0) \Vert_{X \times \mathbf{I}_1} \right)}}
        \arrow[bend right=85, swap]{dddd}[anchor=center, rotate=-90, yshift=-10]{{\Pi \left( e^H e^\phi \right) \Vert (p,i) \Vert_{X \times \mathbf{I}_1} }}
        \arrow[bend left=30]{rrrr}{{\epsilon^{\psi\ast\phi}_{(x,0)}}}
        \& \&
        \Pi \left( e^H \right) \Vert Fx \Vert_Y
        \arrow{dd}{{\Pi \left( e^H \right) \Vert F(p) \Vert_Y}}
        \arrow{rr}{{\epsilon^H_{Fx}}}
        \& \&
        \Vert HFx \Vert_Z
        \arrow{dd}[swap]{\Vert H(F(p)) \Vert_Z}
        \arrow[bend left=40]{dddddd}[anchor=center, rotate=-90, yshift=10]{{\Vert (\psi\ast\phi)(p,i) \Vert_Z}}
        \\ \\
        \Pi \left( e^H \right) \left( e^\phi \Vert (x',0) \Vert_{X \times \mathbf{I}_1} \right)
        \arrow{rr}{{\Pi \left( e^H \right) \left( \epsilon^\phi_{(x',0)} \right)}}
        \arrow{dd}{{\Pi \left( e^H \right) \left( e^\phi \Vert (x',i) \Vert_{X \times \mathbf{I}_1} \right)}}
        \& \&
        \Pi \left( e^H \right) \Vert Fx' \Vert_Y
        \arrow{dd}{{\Pi \left( e^H \right) \left\Vert \phi_{x'} \right\Vert_Y}}
        \arrow{rr}{{\epsilon^H_{Fx'}}}
        \& \&
        \Vert HFx' \Vert_Z
        \arrow{dd}[swap]{\Vert H(F(p)) \Vert_Z}
        \\ \\
        \Pi \left( e^H \right) \left( e^\phi \Vert (x',1) \Vert_{X \times \mathbf{I}_1} \right)
        \arrow{rr}{{\Pi \left( e^H \right) \left( \epsilon^\phi_{(x',1)} \right)}}
        \arrow[bend right=20, swap]{rrrrdd}{{\epsilon^{\phi\ast\psi}_{(x',1)}}}
        \& \&
        \Pi \left( e^H \right) \Vert Gx' \Vert_Y
        \arrow{rr}{{\epsilon^H_{Gx'}}}
        \& \&
        \Vert HGx' \Vert_Z
        \arrow{dd}[swap]{{\left \Vert \psi_{Gx'} \right \Vert_Z}}
        \\ \\
        \& \&
        \& \&
        \Vert KGx' \Vert_Z
    \end{tikzcd}
}
\end{align*}

The identity 2-cell on $F: X \rightarrow Y$ is given by the identity natural transformation, which is realized by
\begin{align*}
    \left( e^F\pi_1, \left(\epsilon^F \ast \pi_1 \right) \circ \left(\Pi\left(e^F\right) \ast {\left\Vert-\right\Vert}_X \pi_1 \right) \right)
\end{align*}
where $(e^F,\epsilon^F)\Vdash F$. The inverse of $\phi: F \Rightarrow G: X \rightarrow Y$ is given by the inverse natural transformation $\phi^{-1}: G \Rightarrow F: X \rightarrow Y$. 
%\begin{align*}
%    &\phi^{-1}: G \Rightarrow F: X \rightarrow Y
%    \\
%    &\phi^{-1}(x,i) \coloneqq (\phi(x,i))^{-1}
%\end{align*}
Assuming $(e^\phi,\epsilon^\phi)$ is a realizer for $\phi$, then $\phi^{-1}$ is realized by $(e^\phi, \widetilde{\epsilon})$, where:
\begin{align*}
    &\widetilde{\epsilon}: \Pi(e^\phi) \circ \Vert \!-\! \Vert_{X \times \mathbf{I}_1} \Rightarrow \Vert \!-\! \Vert_Y \circ \phi^{-1}
    \\
    &\widetilde{\epsilon}_{(x,0)} \coloneqq \epsilon^\phi_{(x,1)}
    \\
    &\widetilde{\epsilon}_{(x,1)} \coloneqq \epsilon^\phi_{(x,0)}
\end{align*}
The axioms for (2,1)-categories hold in virtue of them holding for $\mathbf{Gpd}$.

% Note that $\mathbf{PGAsm}(\mathbf{R},\mathbb{I})$ is a quotient of the comma 2-category $\mathbf{Gpd} \downarrow \Pi$ (a quotient because realizers of (higher) morphisms are only required to exist, not carried around as extra data). Compare this to the $\mathcal{F}$-construction of \cite{robinson01}, which is a set-based, 1-categorical analogue.

%%%%%

As a 2-categorical variation on the theme of realizability categories inheriting (or improving) structure from the realizer category, we will now show that $\mathbf{PGAsm}(\mathbf{R},\mathbb{I})$ is finitely complete as a (2,1)-category whenever $(\mathbf{R},\mathbb{I})$ is (see Remark \ref{rem:squares}). With this hypothesis, we may obtain a realizer $\mathbb{I}_1\times \mathbb{I}_1\rightarrow A$ by providing a commutative square of paths $\mathbb{I}_1 \rightarrow A$. 

\begin{prop}\label{thm:pgasmfinitecomplete}
    If $\partial_\mathbf{R}$ is an isomorphism of double categories then $\mathbf{PGAsm}(\mathbf{R},\mathbb{I})$ is finitely complete as a (2,1)-category. In particular, if $(\mathbf{R},\mathbb{I})$ is finitely complete as a (2,1)-category then so is $\mathbf{PGAsm}(\mathbf{R},\mathbb{I})$.
\end{prop}
\begin{proof}
    A (2,1)-category is finitely complete iff it has a terminal object, pullbacks and pseudopullbacks ((iso)comma objects) \citep{street76}. A 2-category $(\mathbf{C},\mathbb{I})$ arising from an interval has whatever conical (co)limits $\mathbf{C}$ has in the 1-dimensional sense \citep[][Lemma 2.1]{warren12}. A terminal object was exhibited in Proposition \ref{thm:pgasmcartesianclosure}.

    The pullback $F^*Y$ of $G: (Y, B, {\left\Vert-\right\Vert}_Y)\rightarrow (Z,C,{\left\Vert-\right\Vert}_Z)$ along $F: (X,A,{\left\Vert-\right\Vert}_X) \rightarrow Z$ is given by
    \begin{align*}
        \left( F^*Y, A \times B, {\left\Vert-\right\Vert}_{F^* Y} \right)
    \end{align*}
    where $F^*Y$ is the pullback in $\mathbf{Gpd}$ and the realizability functor is defined
    \begin{align*}
        {\left\Vert-\right\Vert}_{F^*Y} \coloneqq \left\langle {\left\Vert \pi_1(-) \right\Vert}_X, {\left\Vert \pi_2(-) \right\Vert}_Y \right\rangle
    \end{align*}
    The projection functors are realized by the respective projections from $\mathbf{R}$. If $S: (W,D,{\left\Vert-\right\Vert}_W) \rightarrow X$ and $T: W \rightarrow Y$ are such that $FS = GT$ then we obtain the universal morphism:
    \begin{align*}
        &[S,T] : W \rightarrow F^*Y
        \\
        &[S,T](-) \coloneqq (S(-),T(-))
    \end{align*}
    that is realized by
    \begin{align*}
        \left( \langle e, e' \rangle: D \rightarrow A \times B, \langle \epsilon, \epsilon' \rangle \right)
    \end{align*}
    where $(e,\epsilon) \Vdash S$ and $(e',\epsilon') \Vdash T$.
    
    The pseudopullback of $G: Y \rightarrow Z$ along $F: X \rightarrow Z$ is given by
    \begin{align*}
        F\downarrow G \coloneqq \left( F\downarrow G, A\times B\times C^{\mathbb{I}_1}, \Vert \!-\! \Vert_{F\downarrow G} \right)
    \end{align*}
    where $F\downarrow G$ is the pseudopullback in $\mathbf{Gpd}$ and the realizability functor is defined:
    \begin{align*}
        {\left\Vert (x, y, r) \right\Vert}_{F\downarrow G} &\coloneqq \left\langle \Vert x \Vert_X, \left\Vert y \right\Vert_Y, \lambda \left\Vert r \right\Vert_Z \right\rangle
        \\
        \left\Vert (p,q) \right\Vert_{F\downarrow G} &\coloneqq \left\langle \left\Vert p \right\Vert_X, \left\Vert q \right\Vert_Y, \lambda\left( \partial^{-1} \left\Vert (p,q,r,r') \right\Vert_Z \right) \right\rangle
    \end{align*}
    where $\Vert (p,q,r,r') \Vert_Z$ denotes the commutative square:
    \begin{align*}
        \begin{tikzcd}[ampersand replacement=\&]
            \left\Vert Fx \right\Vert_Z
            \arrow{r}[yshift=3]{\left\Vert F(p) \right\Vert_Z}
            \arrow{d}[swap, xshift=-3]{\left\Vert r \right\Vert_Z}
            \&
            \left\Vert Fx' \right\Vert_Z
            \arrow{d}[xshift=3]{\left\Vert r' \right\Vert_Z}
            \\
            \left\Vert Gy \right\Vert_Z
            \arrow{r}[swap, yshift=-3]{\left\Vert G(q) \right\Vert_Z}
            \&
            \left\Vert Gy' \right\Vert_Z
        \end{tikzcd}
    \end{align*}
    The projection functors are realized by the respective projections from $\mathbf{R}$.
    
    If we have morphisms $S: W \rightarrow X$ and $T: W \rightarrow Y$ and a 2-cell $\psi:FS \Rightarrow GT$ then we obtain the universal morphism:
    \begin{align*}
        & \left[S,T,\psi\right] : W \rightarrow F\downarrow G
        \\
        &\left[S,T,\psi\right](w) \coloneqq \left( Sw,Tw, \psi(w,i) \right)
        \\
        &\left[S,T,\psi\right](v) \coloneqq \left( S(v),T(v) \right)
    \end{align*}
    Finally, we construct a realizer $(e,\epsilon)\Vdash [S,T,\psi]$ from realizers $(e^S,\epsilon^S) \Vdash S$, $(e^T,\epsilon^T) \Vdash T$ and $(e^\psi,\epsilon^\psi) \Vdash \psi$. For the first component:
    \begin{align*}
        &e: D \rightarrow A \times B \times C^{\mathbb{I}_1}
        \\
        &e \coloneqq \left\langle e^S, e^T, \lambda\left( e^\psi \right) \right\rangle
    \end{align*}
    For the second component:
    \begin{align*}
        \epsilon_w \coloneqq \left\langle \epsilon^S_w, \epsilon^T_w, \lambda \left( \partial^{-1} \left( \epsilon^\psi_{(w,i)} \right) \right) \right\rangle
    \end{align*}
    where $\epsilon^\psi_{(w,i)}$ denotes the following (naturality) square of paths in $\Pi C$.
    \begin{align*}
        \begin{tikzcd}[ampersand replacement=\&]
            e^\psi \left\Vert (w,0) \right\Vert_{W\times \mathbf{I}_1}
            \arrow[rr, "{\epsilon^\psi_{(w,0)}}"]
            \arrow[dd, swap, "{e^\psi \left\Vert (w,i) \right\Vert_{W\times \mathbf{I}_1}}"]
            \& \&
            \substack{\left\Vert \psi(w,0) \right\Vert_Z \\ = \left\Vert FSw \right\Vert_Z}
            \arrow[dd, "{\left\Vert \psi(w,i) \right\Vert_Z}"]
            \\ \\
            e^\psi \left\Vert (w,1) \right\Vert_{W\times \mathbf{I}_1}
            \arrow[rr, swap, "{\epsilon^\psi_{(w,1)}}"]
            \& \&
            \substack{\left\Vert \psi(w,1) \right\Vert_Z \\ = \left\Vert GTw \right\Vert_Z}
        \end{tikzcd}
    \end{align*}
    Double functoriality of $\partial^{-1}$ guarantees the naturality of $\epsilon$.
\end{proof}

\subsection{As path categories}\label{sec:asapathcategory}

Path categories, introduced by \cite{vandenbergm18} \citep[see also][]{vandenberg18a}, are a slight strengthening of Brown's (\citeyear{brown73}) \textit{categories of fibrant objects}, but nevertheless constitute a relatively minimal setting in which to develop abstract homotopy theory as well as model intensional type theory (in general, path categories model a version of type theory in which the computation rule for identity types holds propositionally). A path category $\mathbf{C}$ comes equipped with two classes of maps, fibrations and (weak) equivalences, and satisfies the axioms listed below. An acyclic fibration is a map that is both a fibration and an equivalence. Dependent types are modelled by fibrations.
\begin{enumerate}[label= (PC\arabic*),leftmargin=*,labelindent=1em]
    \item Isomorphisms are fibrations and fibrations are closed under composition.
    \item The pullback of a fibration along any other map exists and is again a fibration.
    \item $\mathbf{C}$ has a terminal object $1$ and every map $X \rightarrow 1$ is a fibration.
    \item Isomorphisms are equivalences.
    \item Equivalences satisfy the 2-out-of-6 property.
    \item Every object in $\mathbf{C}$ has a path object.
    \item Every acyclic fibration has a section.
    \item The pullback of an acyclic fibration along any other map exists and is again an acyclic fibration.
\end{enumerate}
A path object $\mathcal{P}X\in\mathbf{C}$ for an object $X\in\mathbf{C}$ is a factorisation of the diagonal by an equivalence $r$ followed by a fibration $\langle s, t \rangle$.
\begin{align*}
    \begin{tikzcd}[ampersand replacement=\&]
        \&
        \mathcal{P}X
        \arrow[rd, "{\langle s, t \rangle}"]
        \\
        X
        \arrow[ru, "r"]
        \arrow[rr, swap, "\Delta_X"]
        \& \&
        X\times X
    \end{tikzcd}
\end{align*}
Path objects give rise to a nation of (fibrewise) homotopy between morphisms.

Given an object $X\in\mathbf{C}$, let $\mathbf{C}(X)$ be the full subcategory of the slice $\mathbf{C}/X$ spanned by the fibrations. $\mathbf{C}(X)$ is a path category, where a morphism in $\mathbf{C}(X)$ is a fibration or equivalence iff it is so in $\mathbf{C}$. The following result, which is Proposition 2.6 in \citep{vandenbergm18} and is proved on p. 428 of \citep{brown73}, is used in \Cref{sec:impredicativeuniversesofmodestfibrations}.
\begin{lem}[Brown's lemma]\label{thm:brownslemma}
    For any map $f:Y\rightarrow X$, the pullback functor
    \begin{align*}
        f^*: \mathbf{C}(X) \rightarrow \mathbf{C}(Y)
    \end{align*}
    preserves both fibrations and equivalences.
\end{lem}

From hereon in, we assume that $\partial_{\mathbf{R}}$ is an isomorphism of double categories. To make $\mathbf{PGAsm}(\mathbf{R},\mathbb{I})$ into a path category we take fibrations and equivalences to be isofibrations and equivalences respectively, both internal to $\mathbf{PGAsm}(\mathbf{R},\mathbb{I})$ considered as a 2-category. Recall (eg. from \citep[][Section 7.2]{lack10}) that a morphism $p:Y\rightarrow Z$ in a 2-category $\mathbf{C}$ is defined to be an isofibration iff every invertible 2-cell
\begin{align*}
     \begin{tikzcd}[ampersand replacement=\&]
        X
        \arrow[rr, "f"]
        \arrow[ddr, swap, "g" {name=Z}]
        \& \&
        Y
        \arrow[ddl, "p"]
        \\ \\
        \&
        Z
        \arrow[from=1-3, to=Z, Leftarrow, shorten <= 2.6em, shorten >= 1.4em, Rightarrow, yshift=3, near end, "\phi" {xshift=2.5}]
    \end{tikzcd}
\end{align*}
lifts to an invertible 2-cell
\begin{align*}
     \begin{tikzcd}[ampersand replacement=\&]
        X
        \arrow[rr, bend left=55, "f" {name=P}]
        \arrow[rr, swap, "\phi^* f" {name=Q}]
        \arrow[ddr, swap, "g" {name=Z}]
        \& \&
        Y
        \arrow[ddl, "p"]
        \\ \\
        \&
        Z
        \arrow[from=1-3, to=Z, equals, shorten <= 3.0em, shorten >= 2.0em, yshift=-5, xshift=-1]
        \arrow[from=P, to=Q, Rightarrow, shorten <= 0.5em, shorten >= 0.55em, "{\overline{\phi}(f)}" {xshift=3}]
    \end{tikzcd}
\end{align*}
This reduces to the usual notion in the case $\mathbf{C}=\mathbf{Gpd}$.

\begin{lem}
    A morphism in $\mathbf{PGAsm}(\mathbf{R},\mathbb{I})$ is a fibration (internal isofibration) if and only if it is an isofibration in $\mathbf{Gpd}$.
\end{lem}
\begin{proof}
    The forwards direction is trivial. For the backwards direction, suppose we have $F: (X,A,\left\Vert-\right\Vert) \rightarrow (Y,B,\left\Vert-\right\Vert_Y)$ and $\phi: PF \Rightarrow G: X \rightarrow (Z,C,\left\Vert-\right\Vert_Z) \in \mathbf{PGAsm}(\mathbf{R},\mathbb{I})$, where $P: Y \rightarrow Z$ is an isofibration in $\mathbf{Gpd}$. Define the functor:
    \begin{align*}
        &\phi^* F: X \rightarrow Y
        \\
        &\phi^* F (x) \coloneq \left(\phi_x\right)^* Fx
        \\
        &\phi^* F (p:x\rightarrow x') \coloneq \overline{\left(\phi_x\right)}(Fx') \circ F(p) \circ {\overline{\left(\phi_x\right)}(Fx)}^{-1}
    \end{align*}
    where $\left(\phi_x\right)^* Fx$ is the transport of $Fx$ along the path $\phi_x$. So indeed $P \circ \phi^* F = G$. Moreover, define the natural transformation:
    \begin{align*}
        &\overline{\phi}(F): X \times \mathbf{I}_1 \rightarrow Y
        \\
        &\overline{\phi}(F)(p,i) \coloneq \overline{\left(\phi_{x'}\right)}(Fx') \circ F(p)
    \end{align*}
    We have $P \ast \overline{\phi}(F) = \phi$. We need to realize both $\phi^* F$ and $\overline{\phi}(F)$.

    Suppose that $(e,\epsilon)\Vdash F$. We obtain a realizer $(e, \delta) \Vdash \phi^* F$ by setting
    \begin{align*}
        \delta_x \coloneq {\left\Vert \overline{\left(\phi_x\right)} (Fx) \right\Vert}_Y \circ \epsilon_x : \Pi(e)\left\Vert x \right\Vert_X \rightarrow {\left\Vert \phi^* F(x) \right\Vert}_Y
    \end{align*}
    The following diagram encapsulates the naturality of $\delta$.
    \begin{align*}
        \begin{tikzcd}[ampersand replacement=\&]
            \Pi(e)\left\Vert x \right\Vert_X
            \arrow[dd, swap, "\Pi(e)\left\Vert p \right\Vert_X"]
            \arrow[rr, "\epsilon_x"]
            \& \&
            \left\Vert Fx \right\Vert_Y
            \arrow[dd, "\left\Vert F(p) \right\Vert_Y"]
            \arrow[rr, "{\left\Vert \overline{\left(\phi_x\right)} (Fx) \right\Vert}_Y"]
            \& \&
            \left\Vert \phi^* F (x) \right\Vert_Y
            \arrow[dd, "\left\Vert \phi^* F (p) \right\Vert_Y"]
            \\ \\
            \Pi(e)\left\Vert x' \right\Vert_X
            \arrow[rr, swap, "\epsilon_{x'}"]
            \& \&
            \left\Vert Fx' \right\Vert_Y
            \arrow[rr, swap, "{\left\Vert \overline{\left(\phi_x'\right)} (Fx') \right\Vert}_Y"]
            \& \&
            \left\Vert \phi^* F (x') \right\Vert_Y
        \end{tikzcd}
    \end{align*}

    We obtain a realizer $(e\pi_1: A\times \mathbb{I}_1 \rightarrow B, \gamma) \Vdash \overline{\phi}(F)$ by setting:
    \begin{align*}
        \gamma_{(x,0)} &\coloneq \epsilon_x: \Pi(e\pi_1) \left\Vert (x,0) \right\Vert_{X\times \mathbf{I}_1} = \Pi(e)\left\Vert x \right\Vert_X \rightarrow \left\Vert Fx \right\Vert_Y
        \\
        \gamma_{(x,1)} &\coloneq \left\Vert \overline{\left( \phi_x \right)} (Fx) \right\Vert_Y \circ \epsilon: \Pi(e\pi_1) \left\Vert (x,1) \right\Vert_{X\times \mathbf{I}_1} = \Pi(e)\left\Vert x \right\Vert_X \rightarrow \left\Vert \phi^* F (x) \right\Vert_Y
    \end{align*}
    The diagram above also establishes the naturality of $\gamma$, noting that $\Pi(e\pi_1)\left\Vert - \right\Vert_{X\times \mathbf{I}_1} = \Pi(e)\left\Vert \pi_1(-) \right\Vert_X$.
\end{proof}

This is worth emphasizing: given an isofibration $F: (X,A,\left\Vert-\right\Vert) \rightarrow (Y,B,\left\Vert-\right\Vert_Y)$, the fibres $X_y$ are partitioned groupoidal assemblies, where $\left\Vert-\right\Vert_{X_y}$ is just the restriction of $\left\Vert-\right\Vert_X$. Given $q:y \rightarrow y' \in Y$, the transport functor $q^*: X_y \rightarrow X_{y'}$ is realized by $(\mathsf{id}_A, \epsilon)$, where $\epsilon_x \coloneq \left\Vert \overline{q}(x) \right\Vert_X$ (the $F$-lift of $q$ at $x$). This relies on the fact that functors are implemented up to isomorphism.

Using the axiom of choice, every isofibration in $\mathbf{Gpd}$ is equivalent to a split one. The upshot of the following lemma is that any such equivalence can be upgraded to one in $\mathbf{PGAsm}(\mathbf{R},\mathbb{I})$.
\begin{lem}\label{thm:replete}
    Let $X=(X,A,{\left\Vert-\right\Vert}_X) \in \mathbf{PGAsm}(\mathbf{R},\mathbb{I})$ and suppose that we have an equivalence:
    \begin{align*}
        &F: X \rightarrow Y
        &&G: Y \rightarrow X
        &&&\phi: \mathsf{id}_X \Rightarrow GF
        &&&&\psi: \mathsf{id}_Y \Rightarrow FG
    \end{align*}
    in $\mathbf{Gpd}$. Then $Y$ can be equipped with the structure of a partitioned groupoidal assembly in such a way that the above equivalence is elevated to one in $\mathbf{PGAsm}(\mathbf{R},\mathbb{I})$. Moreover, if $X$ is modest then so is $Y$.
\end{lem}
\begin{proof}
    $Y$ is given the realizer type $A$ and the realizability functor
    \begin{align*}
        {\left\Vert-\right\Vert}_Y \coloneqq {\left\Vert-\right\Vert}_X \circ G: Y \rightarrow \Pi A
    \end{align*}
    The functor $G$ is clearly realized by $(\mathsf{id}, \mathsf{id})$. The functor $F$ is realized by $(\mathsf{id}, {\left\Vert-\right\Vert}_X \ast \phi)$. The natural transformation $\phi$ is realized by $(\pi_1,\epsilon^\phi)$, where:
    \begin{align*}
        &\epsilon^\phi_{(x,0)} \coloneqq \mathsf{id}_{{\left\Vert x \right\Vert}_X}
        &&\epsilon^\phi_{(x,1)} \coloneqq {\left\Vert \phi_x \right\Vert}_X
    \end{align*}
    The case for $\psi$ is completely symmetric.%: it is realized by $(\pi_1,\epsilon^\psi)$, where:
    %\begin{align*}
    %    &\epsilon^\psi_{(x,0)} \coloneqq \mathsf{id}_{\Vert y \Vert_Y}
    %    \quad\quad
    %    &\epsilon^\psi_{(x,1)} \coloneqq \left\Vert \psi_y \right\Vert_Y
    %\end{align*}

    The functor $G$, being an equivalence of groupoids, is fully faithful and essentially surjective. So if $X$ is modest then ${\left\Vert-\right\Vert}_Y$ is the composition of two fully faithful functors, and hence fully faithful itself.
\end{proof}

Turning to the path category axioms, (PC1)-(PC3) hold due to standard results about isofibrations (of groupoids), and (PC4) and (PC5) hold for equivalences in any 2-category. We check the remaining axioms below.

\paragraph*{(PC6): path objects}\label{sec:pathobjects}
Given $X=(X,A,{\left\Vert-\right\Vert}_X)$, the weak exponential object $X^{\mathbf{I}_1}$ (see Proposition \ref{thm:pgasmcartesianclosure}) is a path object $\mathcal{P}X$. Thus $\mathcal{P}X$ is modest whenever $X$ is.
    
The equivalence $\mathsf{r}:X\rightarrow X^{\mathbf{I}_1}$ is given by:
\begin{align*}
    r(x) &\coloneqq \left( i \mapsto \mathsf{id}_x, \lambda {\left\Vert \mathsf{id}_x \right\Vert}_X, \mathsf{id} \right)
    \\
    r(p) &\coloneqq \left( (i,i) \mapsto p, \lambda\left( \partial^{-1}{\left\Vert p \right\Vert}_X \right) \right)
\end{align*}
where the argument ${\left\Vert p \right\Vert}_X$ of $\partial^{-1}$ denotes the commutative square whose horizontal edges are ${\left\Vert p \right\Vert}_X$ and whose vertical edges are $\mathsf{id}_x$. $r$ is realized by $(\lambda(\pi_1): A \rightarrow A^{\mathbb{I}_1}, \mathsf{id})$.

The fibration $(s,t)_X: X^{\mathbf{I}_1} \rightarrow X\times X$ is given by:
\begin{align*}
    (F,e,\epsilon) &\mapsto (F0,F1)
    \\
    (\psi,f) &\mapsto \left( \psi(0,i),\psi(1,i) \right)
\end{align*}
realized by $( \langle \mathsf{eval} \circ \langle \mathsf{id},0 \rangle, \mathsf{eval} \circ \langle \mathsf{id},1 \rangle \rangle, \epsilon )$.

Suppose $(F,e,\epsilon)$ is in the fibre over $(x_1,x_2)$. Then $F(i): x_1 \rightarrow x_2$. We define the chosen lift $(\psi,f): (F,e,\epsilon) \rightarrow (G,e,\delta)$ of $p = (p_1,p_2):(x_1,x_2) \rightarrow (x'_1,x'_2)$ at $(F,e,\epsilon)$ as follows. First:
\begin{align*}
    &G(i) \coloneqq p_2 \circ F(i) \circ p_1^{-1}
    &&\delta_0 \coloneqq p_1 \circ \epsilon_0
    &&&\delta_1 \coloneqq p_2 \circ \epsilon_1
\end{align*}
%(compare with (\Cref{eqn:idfunctor}) in \Cref{sec:identitytypes})
Next:
\begin{align*}
    \psi(i,i) \coloneqq p_2 \circ F(i) = G(i) \circ p_1
\end{align*}
Finally:
\begin{align*}
    f \coloneqq \lambda\left( \partial^{-1} \left(\hat{e}\right) \right) : \mathbb{I}_1 \rightarrow A^{\mathbb{I}_1}
\end{align*}
where $\hat{e}$ in the above expression is used to denote the following commutative square of paths in $\Pi A$.
\begin{align*}
    \begin{tikzcd}[ampersand replacement=\&]
        \Pi(\mu(e))0
        \arrow[d, swap, "\Pi(\mu(e))"]
        \arrow[r, equal, "{\delta_0^{-1} p_1 \epsilon_0}"]
        \&
        \Pi(\mu(e))0
        \arrow[d, "\Pi(\mu(e))"]
        \\
        \Pi(\mu(e))1
        \arrow[r, equal, swap, "{\delta_1^{-1} p_2 \epsilon_1}"]
        \&
        \Pi(\mu(e))1
    \end{tikzcd}
\end{align*}

\paragraph*{(PC7)}\label{sec:pc7}

Given an acyclic fibration $F: (X, A, {\left\Vert-\right\Vert}_X ) \rightarrow (Y, B, {\left\Vert-\right\Vert}_Y )$, a pseudoinverse $G: Y \rightarrow X$ of $F$ realized by $(e,\epsilon)$ and 2-cell $\psi: GF \Rightarrow \mathsf{id}_X$, we define a section $S:Y\rightarrow X$ of $F$ by:
\begin{align*}
    Sy &\coloneqq \left(\psi_y\right)^* Gy
    \\
    S(q:y\rightarrow y') &\coloneqq  \overline{\left(\psi_{y'}\right)}(Gy') \circ G(q) \circ \left( \overline{\left(\psi_y\right)}(Gy) \right)^{-1}
\end{align*}
This is realized by $(e, \epsilon')$, where
\begin{align*}
    \epsilon'_y \coloneqq \overline{\left(\psi_{y}\right)}(Gy) \circ \epsilon_y
\end{align*}

\paragraph*{(PC8)}\label{sec:pc8}

Let $G: (Y, B, \left\Vert-\right\Vert_Y) \rightarrow (Z, C, \left\Vert-\right\Vert_Z)$ be an acyclic fibration with pseudoinverse $H: Z \rightarrow Y$ witnessed by natural isomorphisms $\psi: GH \Rightarrow \mathsf{id}_Z$ and $\phi: \mathsf{id}_Y \Rightarrow HG$. Furthermore, let $F: (X, A, {\left\Vert-\right\Vert}_X) \rightarrow Z$ be an arbitrary map. We construct a pseudoinverse $S: X \rightarrow F^*Y$ to the fibration $F^*(G): F^*Y \rightarrow X$. Define $S \coloneqq [X,T]$, where:
\begin{align*}
    &T: X \rightarrow Y
    \\
    &Tx \coloneqq \psi_{Fx}^* HFx
    \\
    &T(p) \coloneqq \overline{\phi_{Fx'}}(HFx') \circ H(F(p)) \circ \left( \overline{\phi_{Fx}}(HFx) \right)
\end{align*}
Given realizers $(e^F,\epsilon^F) \Vdash F$ and $(e^H,\epsilon^H) \Vdash H$, a realizer for $T$ is $(e^Fe^H, \epsilon^T)$, where
\begin{align*}
    \epsilon^T_x \coloneqq \overline{\psi_{Fx}}(HFx) \circ \epsilon^H_{Fx} \circ \Pi\left(e^H\right)\left(\epsilon^F_x\right)
\end{align*}

Clearly we have $F^*G \circ S = \mathsf{id}_X$. We now construct a natural isomorphism $\sigma: \mathsf{id}_{F^*Y} \Rightarrow S \circ F^*(G)$. We have
\begin{align*}
    S\left(F^*(G)(x,y)\right) = Sx = \left( x, \psi^*_{Fx} HFx \right)
\end{align*}
So we define
\begin{align*}
    \sigma_{(x,y)} \coloneqq \left( \mathsf{id}_x, \sigma_y \right)
\end{align*}
where $\sigma_y$ is the following composite.
\begin{align*}
    \begin{tikzcd}[ampersand replacement=\&]
        y
        \arrow[r, "\phi_y"]
        \&
        HGy = HFx
        \arrow[rr, "{\overline{\psi_{Fx}}(HFx)}"]
        \& \&
        \psi^*_{Fx} HFx
    \end{tikzcd}
\end{align*}
This is realized by $(\pi_1: (A\times B)\times \mathbb{I}_1 \rightarrow A \times B, \epsilon^\sigma)$, where:
\begin{align*}
    \epsilon^\sigma_{(x,y,0)} &\coloneqq \left\langle {\left\Vert x \right\Vert}_X, {\left\Vert y \right\Vert}_Y \right\rangle
    \\
    \epsilon^\sigma_{(x,y,1)} &\coloneqq \left\langle {\left\Vert x \right\Vert}_X, {\left\Vert \sigma_y \right\Vert}_Y \right\rangle
\end{align*}

\subsubsection{Dependent products}\label{sec:dependentproducts}

The notion of dependent product for path categories studied as \cite[Definition 5.2][]{vandenbergm18} is as follows. A path category $\mathbf{C}$ is said to have \textit{homotopy $\Pi$-types} iff for any two fibrations $f: X \rightarrow J$ and $\alpha: J \rightarrow I$ there is an object $\Pi_\alpha(F): \Pi_\alpha X \rightarrow I$ in $\mathbf{C}(I)$ together with an "evaluation map" $\mathsf{ev}: \alpha^* \Pi_\alpha X \rightarrow X$ over $J$ with the following universal property: if there is a map $g: Y \rightarrow I$ and a map $m: \alpha^* Y \rightarrow X$ over $J$ then there exists a unique map $n: Y \rightarrow \Pi_\alpha X$ such that $\mathsf{ev} \circ \alpha^* n$  and $m$ are fibrewise homotopic over $J$ (type-theoretically, this is the $\beta$-law). If the uniqueness criterion on $n$ is dropped, we are left with \textit{weak} homotopy $\Pi$-types. Type-theoretically, these correspond to dependent functions that do not necessarily satisfy the $\eta$-law, and thus do not necessarily satisfy function extensionality. If the $\beta$-law holds on the nose then we have (weak) $\Pi$-types. (See \citep{denbesten20} for a more in-depth study of dependent products in path categories.)

\begin{thm}\label{thm:pgasmweakdependentproducts}
    $\mathbf{PGAsm}(\mathbf{R},\mathbb{I})$ has weak $\Pi$-types.
\end{thm}

\begin{proof}
    Given fibrations $G: (X, A, {\left\Vert-\right\Vert}_X) \rightarrow (Y, B, {\left\Vert-\right\Vert}_Y)$ and $F: Y \rightarrow (Z, C, {\left\Vert-\right\Vert}_Z)$, the objects of the underlying groupoid of $\Pi_F X$ are tuples $(z,H,e,\epsilon)$, where $z\in Z$, $H: F\downarrow z \rightarrow X$ in the slice over $Y$
    \begin{align}\label{eqn:depprod}
        \begin{tikzcd}[ampersand replacement=\&]
            F\downarrow z
            \arrow[rr, "H"]
            \arrow[dr, swap, "\pi_1"]
            \& \&
            X
            \arrow[dl, "G"]
            \\
            \&
            Y
        \end{tikzcd}
    \end{align}
    ($F\downarrow z$ is the "homotopy fibre", constructed as a pseudopullback, see Proposition \ref{thm:pgasmfinitecomplete}), $e \in \Pi(A^{B\times C^{\mathbb{I}_1}})$ and $\epsilon$ is a natural isomorphism such that $(\mu(e), \epsilon) \Vdash H$.
    
    Observe that any $r: z \rightarrow z'$ induces a morphism $F\downarrow r: F\downarrow z \rightarrow F\downarrow z'$ of partitioned groupoidal assemblies defined by $(y,u: Fy \rightarrow z) \mapsto (y, ru)$ and identity on morphisms; the morphism is realized by $(e^r, \epsilon^r)$, where:
    \begin{align*}
        e^r &\coloneqq \mathsf{id}_{B\times C^{\mathbb{I}_1}}
        \\
        \epsilon^r_{(y,u)} &\coloneqq \left\langle B, \lambda\left(\partial^{-1} {\left\Vert(r,u)\right\Vert}_Z\right) \right\rangle
    \end{align*}
    and where ${\left\Vert (r,u) \right\Vert}_Z$ denotes the following commutative square.
    \begin{align*}
        \begin{tikzcd}[ampersand replacement=\&]
            {\left\Vert Fy \right\Vert}_Z
            \arrow[r, equal]
            \arrow[d, swap, "{\left\Vert u \right\Vert}_Z"]
            \&
            {\left\Vert Fy \right\Vert}_Z
            \arrow[d, "{\left\Vert ru \right\Vert}_Z"]
            \\
            {\left\Vert z \right\Vert}_Z
            \arrow[r, swap, "{\left\Vert r \right\Vert}_Z"]
            \&
            {\left\Vert z' \right\Vert}_Z
        \end{tikzcd}
    \end{align*}
    Naturality follows double-functoriality of $\partial^{-1}$.
    
    A morphism $(r, \psi, f, \zeta): (z,H,e,\epsilon) \rightarrow (z',H',e',\epsilon')$ in $\Pi_F X$ consists of a morphism $r:z\rightarrow z'$, a natural isomorphism $\psi: H \Rightarrow H' \circ (F\downarrow r)$ over $Y$ and a path $f:e\rightarrow e' \circ e^r = e'$ and a natural isomorphism $\zeta$ satisfying:
    \begin{align*}
        &\zeta(-,0,-) = \epsilon
        &&\zeta(-,1,-) = \left(\epsilon' \ast (F\downarrow r) \right) \circ \left(\Pi(e') \ast \epsilon^r \right)
    \end{align*}
    \begin{align*}
        \begin{tikzcd}[ampersand replacement=\&]
            F\downarrow z
            \arrow[rr, "F\downarrow r"]
            \arrow[dd, swap, "{\left\Vert-\right\Vert}_{F\downarrow z}"]
            \& \&
            F\downarrow z'
            \arrow[dd, "{\left\Vert-\right\Vert}_{F\downarrow z'}"]
            \ar[ddll, Leftarrow, shorten <= 1.25em, shorten >= 1.25em, "\epsilon^r"]
            \arrow[rr, "H'"]
            \& \&
            X
            \arrow[dd, "{\left\Vert-\right\Vert}_X"]
            \ar[ddll, Leftarrow, shorten <= 1.25em, shorten >= 1.25em, "{\epsilon'}"]
            \\ \\
            \Pi\left( B \times C^{\mathbb{I}_1} \right)
            \arrow[rr, equal]
            \& \&
            \Pi (B \times C^{\mathbb{I}_1})
            \arrow[rr, swap, "\Pi(e')"]
            \& \&
            \Pi A
        \end{tikzcd}
    \end{align*}
    as well as
    \begin{align*}
        (\mu(f)\circ\mathsf{swap},\zeta)\Vdash\psi
    \end{align*}
    As in the proof of Proposition \ref{thm:pgasmcartesianclosure}, $\zeta$ is uniquely determined. The realizer type is $C \times A^{B\times C^{\mathbb{I}_1}}$ and the realizability functor is given by
    \begin{align*}
        {\left\Vert-\right\Vert}_{\Pi_F X} \coloneqq \left\langle {\left\Vert-\right\Vert}_Z \circ \pi_1, \pi_3 \right\rangle
    \end{align*}

    The fibration $\Pi_F(G): \Pi_F X \rightarrow Z$ is given by the first projection and realized by $(\pi_1, \mathsf{id})$. The chosen lift of $r$ at $(z,H,e,\epsilon)$ is
    \begin{align*}
        (r, \psi, f): (z,H,e,\epsilon) \rightarrow (z',H',e,\epsilon')
    \end{align*}
    where:
    \begin{align*}
        H' &\coloneqq H \circ \left( F \downarrow r^{-1} \right)
        \\
        \epsilon' &\coloneqq \left(\epsilon \ast F\downarrow r^{-1} \right) \circ \left( e \ast \epsilon^r \right)
        \\
        \psi &\coloneqq \mathsf{id}_H
        \\
        f &\coloneqq \mathsf{id}_e
    \end{align*}
    The third of these definitions makes sense because $(F\downarrow r^{-1}) \circ (F\downarrow r) = \mathsf{id}_{F\downarrow z}$.

    We now define the evaluation map $\mathsf{ev}: F^* \Pi_F X \rightarrow X$. First let us compute $F^*\Pi_F X$. Objects of the underlying groupoid of $F^*\Pi_F X$ are tuples $(y,z,H,e,\epsilon)$, where $y\in Y$, $z=Fy$, and $H$, $e$, $\epsilon$ are as above. A morphism $(q, r, \psi, f): (y,z,H,e,\epsilon) \rightarrow (y',z',H',e',\epsilon')$ consists of morphisms $q: y\rightarrow y'$ and $r=F(q): z \rightarrow z'$, as well as $\psi$, $f$ as described above. The realizer type of $F^*\Pi_F X$ is $B \times C \times A^B$ and the realizability functor is given by
    \begin{align*}
        {\left\Vert-\right\Vert}_{F^*\Pi_F X} \coloneqq \left\langle {\left\Vert-\right\Vert}_Y \circ \pi_1, {\left\Vert-\right\Vert}_Z \circ \pi_2, \pi_4 \right\rangle
    \end{align*}
    Define the map $\mathsf{ev}$ by:
    \begin{align*}
        \mathsf{ev}(y,z,H,e,\epsilon) &\coloneqq H \left( y, \mathsf{id}_z \right)
        \\
        \mathsf{ev}(q,r,\psi,f) &\coloneqq H'(q) \circ \psi_{\left(y,\mathsf{id}_z\right)}
    \end{align*}
    The argument $q$ of $H'$ is here considered as a morphism $(y,r) \rightarrow (y', \mathsf{id}_{z'})$ in $F\downarrow z'$. The map $\mathsf{ev}$ lives in the slice over $Y$ thanks to (\ref{eqn:depprod}). It is realized by $( \mathsf{ev} \circ \langle \pi_3, \pi_1 \rangle, \epsilon' )$, where we overload notation and use $\mathsf{ev}$ for the evaluation map from $\mathbf{R}$, and $\epsilon'_{(y,z,H,e,\epsilon)} \coloneqq \epsilon_y$. This is natural as we know that there is a natural isomorphism $\zeta$ making the following diagram commute.
    \begin{align*}
        \begin{tikzcd}[ampersand replacement=\&]
            \substack{\Pi(\mu(f))  \left\langle 0, {\left\Vert y \right\Vert}_Y \right\rangle \\ = \mu(e) {\left\Vert y \right\Vert}_Y}
            \arrow[dd, "{\Pi(\mu(f)) \left\langle \mathbb{I}_1, {\left\Vert y \right\Vert}_Y \right\rangle}"]
            \arrow[rr, "{\zeta(y,0,i) = \epsilon_y}"]
            \arrow[dddd, bend right=60, swap, "{\Pi(\mu(f))\left\langle \mathbb{I}_1, q \right\rangle}"]
            \& \&
            {\left\Vert H(y,\mathsf{id}_z) \right\Vert}_Y
            \arrow[dd, "{\left\Vert \psi_y \right\Vert}_Y"]
            \\ \\
            \substack{\Pi(\mu(f)) \left\langle 1, {\left\Vert y \right\Vert}_Y \right\rangle \\ = \mu(e'){\left\Vert y \right\Vert}_Y}
            \arrow[rr, "{\zeta(y,1,i) = \epsilon'_y}"]
            \arrow[dd, "\substack{\Pi(\mu(f)) \left\langle \mathbb{I}_1, {\left\Vert q \right\Vert}_Y \right\rangle \\ = \mu(e')\left\Vert q \right\Vert_Y}"]
            \& \&
            \substack{\left\Vert H'\left( F\downarrow r(y, \mathsf{id}_z) \right)  \right\Vert_Y \\ = \left\Vert H'(y,r) \right\Vert_Y}
            \arrow[dd, "\left\Vert H'(q) \right\Vert_Y"]
            \\ \\
            \substack{\Pi(\mu(f)) \left\langle 1, \left\Vert y' \right\Vert_Y \right\rangle \\ = \mu(e')\left\Vert y' \right\Vert_Y}
            \arrow[rr, swap, "\epsilon'_{y'}"]
            \& \&
            {\left\Vert H'\left(y',\mathsf{id}_{z'}\right) \right\Vert}_Y
        \end{tikzcd}
    \end{align*}

    For the universal property, assume we have a map $R: (W, D, \left\Vert-\right\Vert_{W}) \rightarrow Z$ and a map $S: F^*W \rightarrow X$ over $Y$.
    \begin{align*}
        \begin{tikzcd}[ampersand replacement=\&]
            F^*W
            \arrow[rr, "S"]
            \arrow[dr, swap, "F^*R"]
            \& \&
            X
            \arrow[dl, "G"]
            \\
            \&
            Y
        \end{tikzcd}
    \end{align*}
    We construct a map $T$ as in the following diagram.
    \begin{align*}
        \begin{tikzcd}[ampersand replacement=\&]
            W
            \arrow[rr, "T"]
            \arrow[dr, swap, "R"]
            \& \&
            \Pi_F(X)
            \arrow[dl, "\Pi_F(G)"]
            \\
            \&
            Z
        \end{tikzcd}
    \end{align*}
    Define:
    \begin{align*}
        Tw &\coloneqq \left( Rw, H, e, \epsilon \right)
        \\
        T(v:w\rightarrow w') &\coloneqq \left( R(v), \psi, f \right)
    \end{align*}
    where the components $H$, $e$, $\epsilon$, $\psi$ and $f$ are defined below.
    
    First:
    \begin{align*}
        &H: F\downarrow Rw \rightarrow X
        \\
        &H \coloneqq S \circ \left[ \sigma_1, \sigma_2 \right]
    \end{align*}
    Here,
    \begin{align*}
        &\sigma_1: F\downarrow Rw \rightarrow Y
        \\
        &\sigma_1(y,r) \coloneqq r^* y
        \\
        &\sigma_1(q:(y,r)\rightarrow (y',r')) \coloneqq \overline{r'}(y') \circ q \circ \overline{r}(y)^{-1}
    \end{align*}
    is realized by $(\pi_1, \delta)$, where $\delta_y \coloneqq \left\Vert \overline{r}(y) \right\Vert_Y$ and $\sigma_2: F\downarrow Rw \rightarrow W$ is the constantly $w$ functor, realized by $( \left\Vert w \right\Vert_W *, \delta')$, where $\delta_w \coloneqq \mathsf{id}_{\left\Vert w \right\Vert_W}$. We can form $[\sigma_1,\sigma_2]$ because $F(r^*y) = Rw$ and
    \begin{align*}
        F \left( \overline{r'}(y') \circ q \circ \overline{r}(y)^{-1} \right) = r'qr^{-1} = r'qq^{-1}r'^{-1} = \mathsf{id}_{Rw}
    \end{align*}
    
    Next, pick a realizer $(e^S, \epsilon^S) \Vdash S$ and define $e \in \Pi(A^{B\times C^{\mathbb{I}_1}})$ to be the exponential transpose of
    \begin{align*}
        \begin{tikzcd}[ampersand replacement=\&]
            1 \times \left(B \times C^{\mathbb{I}_1}\right)
            \arrow[rr, "{\langle \pi_1\pi_2, \pi_1 \rangle}"]
            \& \&
            B \times 1
            \arrow[rr, "{B \times \left\Vert w \right\Vert_W}"]
            \& \&
            B \times D
            \arrow[r, "e^S"]
            \&
            A
        \end{tikzcd}
    \end{align*}
    The natural isomorphism $\epsilon$ is defined by the following pasting diagram.
    \begin{align*}
        \begin{tikzcd}[ampersand replacement=\&]
            F\downarrow Rw
            \arrow[rr, "{[\sigma_1,\sigma_2]}"]
            \arrow[dd, swap, "\left\Vert-\right\Vert_{F\downarrow Rw}"]
            \& \&
            F^*W
            \arrow[dd, "\left\Vert-\right\Vert_{F^*W}"]
            \ar[ddll, Leftarrow, shorten <= 1.25em, shorten >= 1.25em, "{\left\langle\delta,\delta'\right\rangle}"]
            \arrow[rr, "S"]
            \& \&
            X
            \arrow[dd, "\left\Vert-\right\Vert_X"]
            \ar[ddll, Leftarrow, shorten <= 1.25em, shorten >= 1.25em, "{\epsilon^S}"]
            \\ \\
            \Pi\left( B \times C^{\mathbb{I}_1} \right)
            \arrow[rr, swap, "{\Pi\left( B\times \left\Vert w \right\Vert_W * \right)}"]
            \& \&
            \Pi (B \times D)
            \arrow[rr, swap, "\Pi\left(e^S\right)"]
            \& \&
            \Pi A
        \end{tikzcd}
    \end{align*}
    The $\Pi$-types being constructed are weak because we have had to make a choice of realizer $(e^S,\epsilon^S)$ for each $S$. 
    
    As for $\psi: H \Rightarrow H' \circ (F\downarrow R(v))$, we must have:
    \begin{align*}
        &\psi(y,r,0) = H(y,r) = S\left( r^* y, w \right)
        \\
        &\psi(y',r',1) = H'(y', R(v)r')  = S\left( (R(v)r')^* y', w' \right)
    \end{align*}
    So we define
    \begin{align*}
        \psi(q,i) \coloneqq S\left( \overline{R(v)}\left( (r')^*(y') \right), v \right) \circ H(q) = H'(q) \circ S\left( \overline{R(v)}\left( r^*(y) \right), v \right)
    \end{align*}
    The argument $q$ of $H$ is regarded as a morphism $(y,r)\rightarrow (y',r')$ in $F\downarrow Rw$, whereas $q$ \textit{qua} argument of $H'$ is regarded as a morphism $(y,R(v)r)\rightarrow (y',R(v)r')$ in $F\downarrow Rw'$. Finally, we define $f:e\rightarrow e'$ (recall that $(e,\epsilon) \Vdash H$) to be the exponential transpose of
    \begin{align*}
        \begin{tikzcd}[ampersand replacement=\&]
            \mathbb{I}_1 \times \left( B \times C^{\mathbb{I}_1} \right)
            \arrow[rr, "{\langle \pi_1\pi_2, \pi_1 \rangle}"]
            \& \&
            B \times \mathbb{I}_1
            \arrow[rr, "{B \times \left\Vert v \right\Vert_W}"]
            \& \&
            B \times D
            \arrow[r, "e^S"]
            \&
            A
        \end{tikzcd}
    \end{align*}

    To complete the proof we show that
    \begin{align*}
        S = \mathsf{ev} \circ F^*T
    \end{align*}
    On objects:
    \begin{align*}
        \mathsf{ev}\left( F^*T (y,w) \right)
        &= \mathsf{ev}(y,Tw)
        \\
        &= \mathsf{ev}(y,Rw,H,e,\epsilon)
        \\
        &= H\left(y,\mathsf{id}_{Rw} \right)
        \\
        &= S(y, w)
    \end{align*}
    and on morphisms:
    \begin{align*}
        \mathsf{ev}\left( F^*T (q,v) \right)
        &= \mathsf{ev}(q,T(v))
        \\
        &= \mathsf{ev}(q,R(v),\psi,f)
        \\
        &= H'\left(q:(y,R(v)) \rightarrow \left(y',\mathsf{id}_{Rw'}\right)\right) \circ \psi_{\left( y, \mathsf{id}_{Rw} \right)}
        \\
        &= S\left( q \circ \left(\overline{R(v)}(y)\right)^{-1}, w' \right) \circ S\left( \overline{R(v)}(y), v \right)
        \\
        &= S(q,v)
    \end{align*} 
\end{proof}
The fact that that $S = \mathsf{ev} \circ F^*T$ holds on the nose is down to the fact that isofibrations of groupoids are exponentiable.  

\section{Impredicative universes of modest fibrations}\label{sec:impredicativeuniversesofmodestfibrations}

As we have mentioned, we know from other settings \citep{birkedal00a,birkedal00b,robinson01,lietz02} that impredicativity is intimately related to untypedness at the level of realizers. Therefore, in this section, we work over an untyped realizer category $(\mathbf{R},\mathbb{I},U)$ and in the full subcategory $\mathbf{PGAsm}(\mathbf{R},\mathbb{I},U) \subseteq \mathbf{PGAsm}(\mathbf{R},\mathbb{I})$ spanned by the partitioned groupoidal assemblies whose realizer type is $U$. This subcategory inherits all the categorical structure of $\mathbf{PGAsm}(\mathbf{R},\mathbb{I})$ discussed so far---including the path category structure---as a consequence of the following.
\begin{prop}\label{thm:pgasmequivalence}
    The inclusion
    \begin{align*}
        \mathbf{PGAsm}(\mathbf{R},\mathbb{I},U) \hookrightarrow \mathbf{PGAsm}(\mathbf{R},\mathbb{I})
    \end{align*}
    is an equivalence of 2-categories.
\end{prop}

\begin{proof}
It suffices to show essential surjectivity. For any $X=(X,A,{\left\Vert-\right\Vert}_X)$ we define an isomorphism
\begin{align*}
    F:X\rightarrow X' \coloneq \left(X,U,{\left\Vert-\right\Vert}_{X'}\right)
\end{align*}
where
\begin{align*}
    {\left\Vert-\right\Vert}_{X'} \coloneqq \Pi(s_A) {\left\Vert-\right\Vert}_X
\end{align*}
The underlying functor of $F$ is $\mathsf{id}_X$, which is realized by $(s_A,\mathsf{id})$. The underlying functor of $F^{-1}$ is also $\mathsf{id}_X$, this time realized by $(r_A,\Pi(\rho_A))$.
\end{proof}

In addition to the above, we also allow the underlying groupoids of partitioned groupoidal assemblies to be locally small (objects of $\mathbf{GPD}$). The realizer category is also assumed to be locally small, which means that the fundamental groupoid functor $\Pi:\mathbf{R}\rightarrow \mathbf{Gpd}$ lands in small groupoids, and further, that the category of presheaves $\widehat{\Pi A}$ on any fundamental groupoid is locally small (and thus able to be the underlying groupoid of a partitioned groupoidal assembly).

Van den Berg (\citeyear{vandenberg18b}) defines a notion of \textit{representation} for a subclass $\mathcal{S}$ of fibrations (called "small" fibrations) that contains all isomorphisms and is closed under composition and homotopy pullbacks along arbitrary maps. A representation models a type-theoretic universe where small types are interpreted as small fibrations. The universe is \textit{impredicative} iff the (weak) dependent product $\Pi_\alpha (f) \in \mathcal{S}$ whenever $\alpha$ is a fibration and $f\in\mathcal{S}$. We will now show that $\mathbf{PGAsm}(\mathbf{R},\mathbb{I},U)$ has an impredicative universe of 1-types, where the small types are \textit{modest fibrations}.

% We identify a subclass of the fibrations in $\mathbf{PGAsm}(\mathbf{R},\mathbb{I})$, the modest fibrations, that contain all isomorphisms and are closed under composition and pullbacks along arbitrary maps. We also show that the (weak) dependent product $\Pi_F(M)$ is a modest fibration whenever $M$ is a modest fibration and $F$ is a fibration (not necessarily modest). Finally, we show that there is an impredicative universe of (representation for) modest fibrations: a "generic" modest fibration such that every modest fibration arises up to equivalence as the pullback of the generic one along some map.

\begin{defn}\label{def:modestfibration}
    A modest fibration is a fibration $M:Y\rightarrow X$ such that for all $x:\mathbf{1}\rightarrow X$ the pullback
    \begin{align*}
        \begin{tikzcd}[ampersand replacement=\&]
            Y_x
            \arrow{r}{\pi_2}
            \arrow[d, swap, "x^* M"]
            \arrow[dr, phantom, "\lrcorner", very near start]
            \&
            Y
            \arrow{d}{M}
            \\
            \mathbf{1}
            \arrow{r}[swap]{x}
            \&
            X
        \end{tikzcd}
    \end{align*}
    is modest in the sense of Definition \ref{def:partitionedgroupoidalassemblies}, that is, the realizability functor
    \begin{align*}
        {\left\Vert-\right\Vert}_{Y_x} \coloneqq {\left\Vert-\right\Vert}_Y \circ \pi_2 : Y_x \rightarrow \Pi U
    \end{align*}
    for each of the fibres $Y_x$ is fully faithful.
\end{defn}

\begin{rem}\label{rem:split}
    Suppose that $M:Y \rightarrow X$ is a modest fibration. Is the splitting
    \begin{align*}
        \begin{tikzcd}[ampersand replacement=\&]
            \widetilde{Y}
            \arrow[rr, "S", "\sim"']
            \arrow[dr, swap, "{\widetilde{M}}"]
            \& \&
            Y
            \arrow[dl, "M"]
            \\
            \&
            X
        \end{tikzcd}
    \end{align*}
    of $M$ again a modest fibration? By Brown's lemma it is: Take any $x:\mathbf{1} \rightarrow X$. Then, by Brown's lemma,
    \begin{align*}
        x^* S : \widetilde{Y}_x \rightarrow Y_y
    \end{align*}
    is an equivalence between the fibres. The realizability functor of $\widetilde{Y}_x$ is given by:
    \begin{align*}
        {\left\Vert-\right\Vert}_{\widetilde{Y}_x}
        &\coloneqq {\left\Vert-\right\Vert}_Y \circ S \circ i_{\widetilde{Y}_x}
        \\
        &= {\left\Vert-\right\Vert}_Y \circ i_{Y_x} \circ x^* S
        \\
        &= {\left\Vert-\right\Vert}_{Y_x} \circ x^* S
    \end{align*}
    (where $i$'s denote inclusions of fibres into total groupoids), which is, as the composition of two fully faithful functors, itself fully faithful.
\end{rem}

It is clear that isomorphisms are modest fibrations. By the pullback lemma, modest fibrations are closed under pullback.
\begin{prop}
    Modest fibrations are closed under composition.
\end{prop}
\begin{proof}
    Let $F: X \rightarrow Y$ and $G: Y \rightarrow Z$ be modest fibrations. We want to show that the fibre $X^{GF}_z$ of $GF$ over $z$ is modest for any $z\in Z$.

    For faithfulness, take $x,x'\in X^{GF}_z$ and $\alpha: \left\Vert x \right\Vert_X \rightarrow \left\Vert x' \right\Vert_X \in \Pi U$. We want a morphism $p:x \rightarrow x' \in X^{GF}_z$ such that $\left\Vert p \right\Vert_X = \alpha$. Suppose that $(e,\epsilon) \Vdash F$. Then we have a path
    \begin{align*}
        \begin{tikzcd}[ampersand replacement=\&]
            \left\Vert Fx \right\Vert_Y
            \arrow[r, dashed]
            \&
            \left\Vert Fx' \right\Vert_Y
            \\
            \Pi(e)\left\Vert x \right\Vert_X
            \arrow[u, "\epsilon_x"]
            \arrow[r, swap, "\Pi(e)(\alpha)"]
            \&
            \Pi(e)\left\Vert x' \right\Vert_X
            \arrow[u, swap, "\epsilon_{x'}"]
        \end{tikzcd}
    \end{align*}
    We know that $Fx.Fx'\in Y^G_z$ are both in the fibre of $G$ over $z$, so by fullness of $\left\Vert-\right\Vert_{Y_z}$ we get a morphism $q: Fx \rightarrow Fx' \in Y^G_z$. Now, the transport $q^* x$ and $x'$ are both in the fibre $X^F_{Fx'}$, and we have a path $\alpha \circ \left\Vert \overline{g}(x)^{-1} \right\Vert_X : \left\Vert g^* x \right\Vert_X \rightarrow \left\Vert x' \right\Vert_X$. So by fullness of $\left\Vert-\right\Vert_{X_{Fx'}}$ we get a morphism $r: g^* x \rightarrow x' \in X^F_{Fx'}$. So we take $p \coloneqq r \circ \overline{q}(x)$. Clearly,
    \begin{align*}
        \left\Vert p \right\Vert_X = \left\Vert r \right\Vert_X \circ \left\Vert \overline{q}(x) \right\Vert_X = \alpha \circ \left\Vert \overline{q}(x)^{-1} \right\Vert_X \circ \left\Vert \overline{q}(x) \right\Vert_X = \alpha
    \end{align*}
    and
    \begin{align*}
        G(F(p)) = G \left(F\left(r\circ \overline{q}(x) \right) \right) = G(q) = z
    \end{align*}

    For faithfulness, suppose $p,q: x \rightarrow x' \in X^{GF}_z$ such that $\left\Vert p \right\Vert_{X} = \left\Vert q \right\Vert_{X}$. We want to show that $p=q$. We know that $F(p), F(q)\in Y^G_z$. So we know that $\left\Vert F(p) \right\Vert_Y = \left\Vert F(q) \right\Vert_Y$.
    \begin{align*}
            \begin{tikzcd}[ampersand replacement=\&]
            \Pi(e)\left\Vert x \right\Vert_X
            \arrow[r, "\epsilon_x"]
            \arrow[d, swap, "\substack{\Pi(e)\left\Vert p \right\Vert_X \\ = \Pi(e)\left\Vert q \right\Vert_X}"]
            \&
            \left\Vert Fx \right\Vert_Y
            \arrow[d, swap, "\left\Vert F(p) \right\Vert_Y", xshift=-6]
            \arrow[d, "\left\Vert F(q) \right\Vert_Y", xshift=4]
            \\
            \Pi(e)\left\Vert x' \right\Vert_X
            \arrow[r, swap, "\epsilon_{x'}"]
            \&
            \left\Vert Fx' \right\Vert_Y
        \end{tikzcd}
    \end{align*}
    By faithfulness of $\left\Vert-\right\Vert_{Y_z}$ we infer that $F(p)=F(q)\eqcolon r$. Then $r^* x, x' \in X^F_{Fx'}$ and we have maps
    \begin{align*}
        &p \circ {\overline{r}(x)}^{-1}
        &&q \circ {\overline{r}(x)}^{-1}
    \end{align*}
    such that
    \begin{align*}
        \left\Vert p \circ {\overline{r}(x)}^{-1} \right\Vert_X = \left\Vert q \circ {\overline{r}(x)}^{-1} \right\Vert_X
    \end{align*}
    By faithfulness we deduce $p \circ {\overline{r}(x)}^{-1} = q \circ {\overline{r}(x)}^{-1}$ and thus $p=q$.
\end{proof}

We will now explain what a representation for modest fibrations amounts to. A representation $\theta$ for modest fibrations is a modest fibration $\theta:\Theta\rightarrow\Lambda$ such that for every modest fibration $M:Y \rightarrow X$ there is a map $M_*: Y \rightarrow \Theta$ such that $\theta \circ M_* = \nu_M \circ M$ and such that the induced map $\left[ M_*, M \right] : Y \rightarrow P$ is an equivalence $M \simeq (\nu_M)^*\theta$ in $\mathbf{PGAsm}(\mathbf{R},\mathbb{I},U)$. 
\begin{align*}
    \begin{tikzcd}[ampersand replacement=\&]
        Y
        \arrow[ddrr, bend right=18, swap, "M"]
        \arrow[rr, dashed,  "{\left[ M_*, M \right]}", "{\sim}"']
        \arrow[rrrr, bend left=42, "{M_*}"]
        \& \&
        P
        \arrow{rr}
        \arrow{dd}[swap]{(\nu_M)^*\theta}
        \arrow[ddrr, phantom, "\lrcorner", very near start]
        \& \&
        \Theta
        \arrow{dd}{\theta}
        \\ \\
        \& \&
        X
        \arrow{rr}[swap]{\nu_M}
        \& \&
        \Lambda
    \end{tikzcd}
\end{align*}

\begin{thm}\label{thm:representationmodestfibrations}
    $\mathbf{PGAsm}(\mathbf{R},\mathbb{I},U)$ has a representation for modest fibrations.
\end{thm}
Before giving the proof, we first define the "chaotic" inclusion
\begin{align*}
    \nabla:\mathbf{GPD} \rightarrow\mathbf{PGAsm}(\mathbf{R},\mathbb{I},U)
\end{align*}
of groupoids into partitioned groupoidal assemblies. Choose an arbitrary $a_0\in \Pi U$ and set:
\begin{align*}
    \nabla X &\coloneqq \left( X, {\left\Vert-\right\Vert}_{\nabla X} : p \mapsto a_0 \right)
    \\
    \nabla(F) &\coloneqq F
\end{align*}
where $\nabla(F)$---in fact any functor into an object in the image of $\nabla$---is realized by $(a_0*,\mathsf{id})$ (ie. $\nabla$ is right adjoint to the underlying groupoid functor). 
\begin{proof}[Proof of Theorem \ref{thm:representationmodestfibrations}]
    Define
    \begin{align*}
        \Lambda \coloneqq \nabla \widehat{\Pi U}
    \end{align*}
    where $\widehat{\Pi U}$ is the groupoid of $\mathbf{Bij}$-valued presheaves on $\Pi U$.
    
    The underlying groupoid of $\Theta$ has as objects pairs $(F, a)$, where $F\in\widehat{\Pi U}$ and $a\in \Pi U$ is such that $Fa\neq\emptyset$. A morphism $(\psi,\alpha): (F,a)\rightarrow (G,b)$ consists of a natural isomorphism $\psi:F\Rightarrow G$ and a path $\alpha: a\rightarrow b$. The realizability functor ${\left\Vert-\right\Vert}_\Theta$ is given by the second projection. The modest fibration $\theta\coloneqq\pi_1:\Theta\rightarrow\Lambda$ is given by the first projection. Fix $(F,a)$ and $(F,b)$ in the fibre over $F\in\Lambda$. Every $\alpha:a\rightarrow b$ in $\Pi U$ uniquely determines a morphism $(\mathsf{id}_F,\alpha)$ in the fibre over $F$. The chosen lift of $\psi:F\rightarrow G$ at $(F,a)$ is $(\psi,\mathsf{id}_a):(F,a) \rightarrow (G,a)$.
    
    Let $M:Y\rightarrow X$ be a modest fibration. By Remark \ref{rem:split}, we can (using AC) take $M$ to be split. Define the object part of its characteristic map $\nu_M:X \rightarrow \Lambda$ of $M$ as follows. First,
    \begin{align*}
        &\faktor{\nu_M(x)(a) \coloneqq \left\{ \alpha:a\rightarrow a' \mid \exists y\in Y. \, My=x \land \Vert y \Vert_Y = a' \right\}}{\sim}
    \end{align*}
    where the relation $\sim$ is isomorphism in the coslice above $a$. Then
    \begin{align*}
        &\nu_M(x)(\beta:a\rightarrow b)[\alpha] \coloneqq [\alpha \circ \beta^{-1}]
    \end{align*}

    To define $\nu_M(p:x\rightarrow x')(\beta,i)[\alpha]$, we pick a representative $\alpha:a\rightarrow c$ from $[\alpha]$. Further, we pick $y_x^{c}\in Y$ such that $M(y_x^{c})=x$ and $\left\Vert y_x^{c} \right\Vert_Y = c$. Using the fact that $M$ is an isofibration, we define
    \begin{align*}
        \nu_M(p:x\rightarrow x')(\beta,i)[\alpha] \coloneqq \left[ {\left\Vert \overline{p}\left( y_x^{c} \right) \right\Vert}_Y \circ \alpha \circ \beta^{-1} \right]
    \end{align*}
    The following diagram helps visualize this definition.
    \begin{align*}
        \begin{tikzcd}[ampersand replacement=\&]
            \Pi U
            \& \&
            b
            \arrow[rr, "\beta^{-1}"]
            \& \&
            a
            \arrow[rr, "\alpha"]
            \& \&
            c
            \arrow[rr, "\left\Vert \overline{p}\left( y^{c}_{x} \right) \right\Vert_Y"]
            \& \&
            {\left\Vert p^* y^{c}_{x} \right\Vert}_Y
            \\ \\
            Y
            \arrow[uu, "{\left\Vert-\right\Vert}_Y"]
            \arrow[dd, swap, "M"]
            \& \& \& \& \& \&
            y^{c}_{x}
            \arrow[rr, "\overline{p}\left( y^{c}_{x} \right)"]
            \& \&
            p^* y^{c}_{x}
            \\ \\
            X
            \& \& \& \& \& \&
            x
            \arrow[rr, swap, "p"]
            \& \&
            x'
        \end{tikzcd}
    \end{align*}

    This definition is independent of the choices made: Let $\alpha':a\rightarrow d$ be a element of $[\alpha]$ isomorphic to $\alpha$ in the coslice above $a$, and let $y^d_x \in Y$ be a choice of element such that $M(y_x^{d})=x$ and ${\left\Vert y_x^{d} \right\Vert}_Y = d$. Then by modesty of $M$ we have an isomorphism $y_x^{c} \rightarrow y_x^{d}$. Transporting this isomorphism along $p$ and applying ${\left\Vert-\right\Vert}_Y$ yields an isomorphism
    \begin{align*}
        \left\Vert \overline{p}\left( y_x^{c} \right) \right\Vert_Y \circ \alpha \circ \beta^{-1} \rightarrow \left\Vert \overline{p}\left( y_x^{d} \right) \right\Vert_Y \circ \alpha \circ \beta^{-1}
    \end{align*}
    in the coslice above $b$.

    To show that $\nu_M(x)$ is functorial, first note that $\nu_M(\mathsf{id}_x)(\beta,i)[\alpha] = [\alpha \circ \beta^{-1}]$. For composition, we exhibit an isomorphism between
    \begin{align*}
        \left\Vert \overline{qp}\left( y_x^{c} \right) \right\Vert_Y \circ \alpha \circ \beta^{-1}
    \end{align*}
    and
    \begin{align*}
        \left\Vert \overline{p}\left( y_{x'}^{\left\Vert p^* y^c_x \right\Vert_Y} \right) \right\Vert_Y \circ \left\Vert \overline{p}\left( y_x^{c} \right) \right\Vert_Y \circ \alpha \circ \beta^{-1}
    \end{align*}
    in the slice over $b$, where $y_{x'}^{\left\Vert p^* y^c_x \right\Vert_Y}$ is a chosen element such that:
    \begin{align*}
    M\left(y_{x'}^{\left\Vert p^* y^c_x \right\Vert_Y}\right) &= x'
    \\
    \left\Vert y_{x'}^{\left\Vert p^* y^c_x \right\Vert_Y} \right\Vert_Y &= \left\Vert p^* y^c_x \right\Vert_Y
    \end{align*}
    Given that $p^* y^c_{x'}$ and $y_{x'}^{\left\Vert p^* y^c_x \right\Vert_Y}$ are both in the fibre over $x'$ and the image of each under ${\left\Vert-\right\Vert}_Y$ is $\left\Vert p^* y^c_x \right\Vert_Y$, by the fullness of ${\left\Vert-\right\Vert}_Y$ we obtain an isomorphism
    \begin{align*}
        \left\Vert p^* y^c_x \right\Vert^{-1}_Y : p^* y^c_{x'} \rightarrow y_{x'}^{\left\Vert p^* y^c_x \right\Vert_Y}
    \end{align*}
    Transporting this isomorphism along $q$ and then applying ${\left\Vert-\right\Vert}_Y$ gives the desired isomorphism.
    
    The pullback $P \coloneqq X \times_\Lambda \Theta$ has objects of the form $(x,a)$ (we omit the component $\nu_M(x)$ that is determined by $x$) and morphisms of the form $(p,\alpha)$. To show that $M$ is equivalent to $(\nu_M)^*\theta$ we first construct $M_*:Y\rightarrow \Theta$.
    \begin{align*}
        M_*(y) &\coloneqq \left( \nu_M(My), \left\Vert y \right\Vert_Y \right)
        \\
        M_*(q) &\coloneqq \left( \nu_M(M(q)), \left\Vert q \right\Vert_Y \right)
    \end{align*}
    This is realized by $(\mathsf{id},\mathsf{id})$ and indeed satisfies $\nu_M \circ M = \theta \circ M_*$. So we get the universal map $[M_*,M]$.
    
    Now we define a pseudoinverse $M^*:P\rightarrow Y$ to $[M_*,M]$. Given $(x,a)\in P$, we choose an element $[\gamma_{x,a}] \in \nu_M(x)(a)$ and a representative $\gamma_{x,a}: a \rightarrow c_{x,a}$ from $[\gamma_{x,a}]$ (we know $\nu_M(x)(a)$ is non-empty). Moreover, $\upsilon(\gamma_{x,a}) \in Y$ is a chosen element such that $M(\upsilon(\gamma_{x,a})) = x$ and $\left\Vert \upsilon(\gamma_{x,a}) \right\Vert_Y = c_{x,a}$.
    
    With this, on objects we define
    \begin{align*}
        M^*\left(x,a\right) \coloneqq \upsilon\left(\gamma_{x,a}\right)
    \end{align*}
    Now take a morphism $\left(p,\alpha\right): \left(x,a\right)\rightarrow \left(x',b\right) \in P$. We would like to define its image under $M^*$ to be
    \begin{align}\label{eqn:M*morphismsalmost}
        \overline{p}\left( \upsilon\left(\gamma_{x,a}\right) \right): \upsilon\left(\gamma_{x,a}\right) \rightarrow p^* \upsilon\left(\gamma_{x,a} \right)
    \end{align}
    but it is not necessarily the case that
    \begin{align}\label{eqn:codomainsalign}
        p^* \upsilon\left(\gamma_{x,a} \right) = \upsilon\left(\gamma_{x',b}\right)
    \end{align}
    ie. the codomain may not align. Here we utilise modesty. Consider the following commutative diagram in $\Pi U$.
    \begin{align*}
        \begin{tikzcd}[ampersand replacement=\&]
            a
            \arrow[rr, "\alpha"]
            \arrow[d, swap, "\gamma_{x,a}"]
            \& \&
            b
            \arrow[d, "\gamma_{x',b}"]
            \\
            c_{x,a}
            \arrow[dr, swap, "{\left\Vert \overline{p}\left( \upsilon\left(\gamma_{x,a}\right) \right) \right\Vert_Y }"]
            \& \&
            c_{x',b}
            \\
            \&
            \left\Vert p^* \upsilon\left(\gamma_{x,a}\right) \right\Vert_Y
            \arrow[ur, dashed, swap, "\delta_{p,\alpha}"]
        \end{tikzcd}
    \end{align*}
    We know that
    \begin{align*}
        p^* \upsilon(\gamma_{x,a}), \upsilon(\gamma_{x',b}) \in Y_{x'}
    \end{align*}
    and we have the morphism
    \begin{align*}
        \delta_{p,\alpha}: \left\Vert p^* \upsilon\left(\gamma_{x,a} \right) \right\Vert_Y \rightarrow c_{x',b} = \left\Vert \upsilon\left(\gamma_{x',b}\right) \right\Vert_Y
    \end{align*}
    Thus, using fullness, we obtain a morphism
    \begin{align*}
        \left\Vert \delta_{p,\alpha} \right\Vert^{-1}_Y : p^* \upsilon\left(\gamma_{x,a} \right) \rightarrow \upsilon\left(\gamma_{x',b}\right)
    \end{align*}
    in the fibre $Y_{x'}$. Therefore we can define
    \begin{align*}
        M^*\left(p,\alpha\right) \coloneqq {\left\Vert \delta_{p,\alpha} \right\Vert}^{-1}_Y \circ \overline{p}\left( \upsilon\left(\gamma_{x,a}\right) \right)
    \end{align*}
    which reduces to (\ref{eqn:M*morphismsalmost}) in case (\ref{eqn:codomainsalign}) holds. Faithfulness ensures this is functorial. $M^*$ is realized by $(\pi_2 \circ r_{U\times U}, \epsilon)$, where $\epsilon_{(x,a)} \coloneqq \gamma_{x,a} \circ \rho_a$.

    There is a natural isomorphism $\sigma: \mathsf{id}_Y \Rightarrow M^* [M_*,M]$, defined
    \begin{align*}
        \sigma(q,i) \coloneqq \left\Vert \gamma_{My',\left\Vert y' \right\Vert_Y} \right\Vert^{-1}_Y \circ q 
    \end{align*}
    and realized by $( \pi_1 r_{U\times U}, \epsilon)$, where:
    \begin{align*}
        \epsilon_{(y,0)} &\coloneqq \left( \left\langle {\left\Vert-\right\Vert}_Y , {\left\Vert-\right\Vert}_{\mathbf{I}_1} \right\rangle \ast \Pi\left( \rho_{U\times U} \right) \ast \Pi(\pi_1) \right)_{(y,0)}
        \\
        \epsilon_{(y,1)} &\coloneqq \gamma_{My,{\left\Vert y \right\Vert}_Y} \circ \epsilon_{(y,0)}
    \end{align*}
    Conversely, there is a natural isomorphism $\tau: \mathsf{id}_P \Rightarrow [M_*,M] M^*$, defined
    \begin{align*}
        \tau\left(\left(p, \alpha \right), i\right) \coloneqq \left( p, \gamma_{x',b} \circ \alpha \right)
    \end{align*}
    and realized by $( \pi_1 r_{U\times U}, \epsilon')$, where:
    \begin{align*}
        \epsilon'_{\left( \left(x, a \right), 0 \right)} &\coloneqq \left( \left\langle {\left\Vert-\right\Vert}_P , {\left\Vert-\right\Vert}_{\mathbf{I}_1} \right\rangle \ast \Pi\left( \rho_{U\times U} \right) \ast \Pi(\pi_1) \right)_{\left( \left(x, a \right), 0 \right)}
        \\
        \epsilon'_{\left( \left(x, a \right), 1 \right)} &\coloneqq \left( \mathsf{id}_x, \mathsf{id}_{\nu_M(x)}, \gamma_{x',b} \right) \circ \epsilon'_{\left( \left(x, a \right), 0 \right)}
    \end{align*}
\end{proof}

To round off the impredicative universe, we have:
\begin{thm}
    Let $F: Y \rightarrow Z$ be a fibration and $G: X \rightarrow Y$ be a modest fibration. Then the dependent product $\Pi_F(G): \Pi_F X \rightarrow Z$ is a modest fibration. 
\end{thm}
\begin{proof}
    The argument is similar to that at the end of the proof of Proposition \ref{thm:pgasmcartesianclosure}. Given $(z,H,e,\epsilon), (z,H',e',\epsilon') \in (\Pi_F X)_z$ and $(y,u)\in F\downarrow z$, we know that $H(y,u)$ and $H'(y,u)$ are both in the fibre $X^G_y$ by the commutativity of (\ref{eqn:depprod}). Then any $f:e \rightarrow e'$ uniquely determines a morphism $(\mathsf{id}_z, \psi, f) \in (\Pi_F X)_z$ because ${\left\Vert-\right\Vert}_{X_y}$ is fully faithful: the image of the component $\psi_y$ under ${\left\Vert-\right\Vert}_{X_y}$ is the unique morphism making the following square commute.
    \begin{align*}
        \begin{tikzcd}[ampersand replacement=\&]
            {\left\Vert Hx \right\Vert}_{X_y}
            \arrow[rr, dashed, "{\left\Vert \psi_y \right\Vert}_{X_y}"]
            \& \&
            {\left\Vert H'x \right\Vert}_{X_y}
            \\ \\
            \substack{\Pi(f) \langle 0, {\left\Vert y \right\Vert}_{X_y} \rangle \\ = \Pi(e){\left\Vert y \right\Vert}_{X_y}}
            \arrow[uu, "\epsilon_y"]
            \arrow[rr, swap,  "{\Pi(f) \langle \mathbb{I}_1, {\left\Vert y \right\Vert}_{X_y} \rangle}"]
            \& \&
            \substack{\Pi(f) \langle 1, {\left\Vert y \right\Vert}_{X_y} \rangle \\ = \Pi(e'){\left\Vert y \right\Vert}_{X_y}}
            \arrow[uu, swap, "\epsilon'_y"]
        \end{tikzcd}
    \end{align*}
\end{proof}

\section{Outlook}\label{sec:outlook}

We have exhibited a model of 1-truncated intensional type theory with an impredicative universe of 1-types in the category $\mathbf{PGAsm}(\mathbf{R},\mathbb{I},U)$ of partitioned groupoidal assemblies over the untyped realizer category $(\mathbf{R},\mathbb{I},U)$. This generalizes set-based partitioned assemblies over a cartesian closed category. More broadly, we have opened the door to a categorical treatment of higher-dimensional realizability---where realizers themselves carry higher-dimensional structure. One could consider alternative classes of fibrations in $\mathbf{PGAsm}(\mathbf{R},\mathbb{I},U)$. For instance, one might take fibrations to be given by the homotopy lifting property, which would be tantamount to the lifting operation being realized.

Ultimately, we would like to investigate groupoidal assemblies (not necessarily partitioned) and groupoidal realizability toposes. These could be studied directly or as free completions of partitioned groupoidal assemblies. \cite{shulman21} has studied 2-dimensional regular and exact completions. One can take the regular or exact completion of any finitely complete 2-category.
\begin{itemize}
    \item Is the exact completion of $\mathbf{PGAsm}(\mathbf{R},\mathbb{I},U)$ an elementary (non-Grothendieck) (2,1)-topos?
\end{itemize}

To be sure, we do not seek a 1-topos but a (2,1)-topos; and it is 2-dimensional regular and exact completions that are relevant. A result due to \cite{lumsdaine11} states that in any coherent 1-category, any cocategory is a co-equivalence relation, ie. it is a cogroupoid whose endpoint maps are jointly epimorphic---in terms of fundamental groupoids, this means that any two parallel paths are equal. We escape this limitation in the higher setting: for example, the (2,1)-category $\mathbf{Gpd}$ is a (2,1)-topos but contains a cogroupoid that is not a co-equivalence relation (Example \ref{eg:intervalgroupoid}). On the other hand, \cite{vandenbergm18} study the "homotopy exact completion" $\mathbf{Hex}(\mathbf{C})$ of a path category $\mathbf{C}$, which turns out to be equivalent to the ex/lex completion ${\mathbf{Ho}(\mathbf{C})}_{\mathsf{ex/lex}}$ of the homotopy category $\mathbf{Ho}(\mathbf{C})$ of $\mathbf{C}$ (their Proposition 3.18).

We would like to know what principles are valid in these putative models.
\begin{itemize}
    \item Do they contain impredicative and univalent universes?
    \item What about propositional resizing?
    \item Church's thesis?
\end{itemize}
A model with an impredicative universe and function extensionality contains refined encodings of (higher) inductive types that satisfy their full universal property \citep{awodey18}.

Beyond this, we would like to find more examples of untyped realizer categories, or else generalize untyped realizer categories to admit more examples. In particular, we would like to do groupoidal realizability over higher-dimensional untyped $\lambda$-calculi, such as a cubical untyped $\lambda$-calculus (roughly, an untyped version of a much stripped back version of cubical type theory \cite{cohen15}). We have also begun thinking about partial computable functors over the groupoid $\mathbf{FinSet}$ of finite sets and bijections (beefed-up natural numbers).

Finally, we would eventually like to investigate weak $\infty$-groupoidal realizability. This should use a weaker notion of interval compared with that used here, in that the cogroupoid axioms should be allowed to hold up to homotopy; it should accomodate $[0,1] \in \mathbf{Top}$---with no quotienting by homotopy. For further discussion of future work see \citep{speight23}.

\paragraph*{Acknowledgements}
Thanks to Samson Abramsky, Carlo Angiuli, Steve Awodey, Robert Harper, Andrzej Murawski, Michael Shulman, Benno van den Berg and Michael Warren for helpful discussions.

The research presented here was carried out during the author's DPhil at the University of Oxford, under the supervision of Samson Abramsky. For some of that time the author was supported by an EPSRC Doctoral Training Partnership Studentship. The author's affiliation changed to University of Birmingham before preparation of the paper.

%%%%%


\begin{thebibliography}{}

    \bibitem[\protect\citename{Abramsky, }1995]{abramsky95}
    Abramsky, S. (1995). \textit{Typed realizability}. Talk at the workshop on Category Theory and Computer Science in Cambridge, England.

    \bibitem[\protect\citename{Adámek, }1997]{adamek97}
    Adámek, J. (1997). A categorical generalization of Scott domains. \textit{Math. Struct. Comput. Sci.}, 7(5):419–443. \href{https://doi.org/10.1017/S0960129597002351}{doi:10.1017/S0960129597002351}.

    \bibitem[\protect\citename{Angiuli {\it et al.}, }2017]{angiuli17}
    Angiuli, C., Harper, R., and Wilson T. (2017). Computational higher-dimensional type theory. In Giuseppe Castagna and Andrew D. Gordon, editors, \textit{Proceedings of the 44th ACM SIGPLAN Symposium on Principles of Programming Languages, POPL 2017}, Paris, France, January 18-20, 2017, pages 680–693. ACM. \href{https://doi.org/10.1145/3009837.3009861}{doi:10.1145/3009837.3009861}.

    \bibitem[\protect\citename{Angiuli and Harper, }2017]{angiuli17meaning}
    Angiuli, C. and Harper, R. (2017). Meaning explanations at higher dimension. \textit{Indagationes Mathematicae}, 29, 10. \href{https://doi.org/10.1016/j.indag.2017.07.010}{doi:10.1016/j.indag.2017.07.010}.

    \bibitem[\protect\citename{Angiuli {\it et al.}, }2018]{angiuli18}
    Angiuli, C., Hou, K. (Favonia), and Harper, R. (2018). Cartesian cubical computational type theory: Constructive reasoning with paths and equalities. In Dan R. Ghica and Achim Jung, editors, \textit{27th EACSL Annual Conference on Computer Science Logic, CSL 2018, September 4-7, 2018, Birmingham, UK, volume 119 of LIPIcs}, pages 6:1–6:17. Schloss Dagstuhl - Leibniz-Zentrum für Informatik. \href{https://doi.org/10.4230/LIPIcs.CSL.2018.6}{doi:10.4230/LIPIcs.CSL.2018.6}.

    \bibitem[\protect\citename{Awodey {\it et al.}, }2018]{awodey18}
    Awodey, S., Frey, J., Speight, S. (2018). Impredicative encodings of (higher) inductive types. \textit{In Proceedings of the 33rd Annual ACM/IEEE Symposium on Logic in Computer Science (LICS '18). Association for Computing Machinery, New York, NY, USA}, 76–85. \href{https://doi.org/10.1145/3209108.3209130}{doi:10.1145/3209108.3209130}.

    \bibitem[\protect\citename{Bauer, }2012]{bauer12}
    Bauer, A. (2012). Intuitionistic mathematics and realizability in the physical world. In H. Zenil (ed.), {\it A Computable Universe: Understanding and Exploring Nature as Computation}, pp.~143--57. World Scientific. \href{https://doi.org/10.1142/9789814374309\_0008}{doi:10.1142/9789814374309\_0008}.

    \bibitem[\protect\citename{Bauer, }2022]{bauer22}
    Bauer, A (2022). \textit{Notes on realizability}. Available online: \url{https://www.andrej.com/zapiski/MGS-2022/notes-on-realizability.pdf}.

    \bibitem[\protect\citename{Birkedal, }1999]{birkedal99}
    Birkedal, L. (1999). \textit{Developing Theories of Types and Computability via Realizability}. PhD thesis, Carnegie Mellon University.

    \bibitem[\protect\citename{Birkedal, }2000a]{birkedal00a}
    Birkedal, L. (2000). Developing theories of types and computability via realizability. \textit{Electronic Notes in Theoretical Computer Science}, 34, 06. \href{https://doi.org/10.1016/S1571-0661(05)80642-5}{doi:10.1016/S1571-0661(05)80642-5}.

    \bibitem[\protect\citename{Birkedal, }2000b]{birkedal00b}
    Birkedal, L. (2000). A general notion of realizability. In \textit{15th Annual IEEE Symposium on Logic in Computer Science, Santa Barbara, California, USA, June 26-29, 2000}, pages 7–17. IEEE Computer Society.

    \bibitem[\protect\citename{Brown, }1973]{brown73}
    Brown, K. S. (1973). Abstract homotopy theory and generalized sheaf cohomology. \textit{Transactions of the American Mathematical Society}, 186:419–458. \href{https://doi.org/10.2307/1996573}{doi:10.2307/1996573}.

    \bibitem[\protect\citename{Cavallo, }2019]{cavallo19}
    Cavallo, E. and Harper, R. (2019). Higher inductive types in cubical computational type theory. \textit{Proc. ACM Program. Lang.}, 3(POPL):1:1–1:27. \href{https://doi.org/10.1145/3290314}{doi:10.1145/3290314}.

    \bibitem[\protect\citename{Cohen {\it et al.}, }2015]{cohen15}
    Cohen, C., Coquand, T., Huber, S., Mörtberg, A. (2015). A Constructive Interpretation of the Univalence Axiom. In {\it 21st International Conference on Types for Proofs and Programs (TYPES 2015). Leibniz International Proceedings in Informatics (LIPIcs)}, Volume 69, pp. 5:1-5:34, Schloss Dagstuhl - Leibniz-Zentrum für Informatik. \href{https://doi.org/10.4230/LIPIcs.TYPES.2015.5}{10.4230/LIPIcs.TYPES.2015.5}

    \bibitem[\protect\citename{Curry, }1930]{curry30}
    Curry, H. B. (1930). Grundlagen der kombinatorischen logik. \textit{American Journal of Mathematics}, 52(4):789–834. \href{https://doi.org/10.2307/2370619}{doi:10.2307/2370619}.

    \bibitem[\protect\citename{den Besten, }2020]{denbesten20}
    den Besten, M. (2020). On homotopy exponentials in path categories. \textit{arXiv: Category Theory.} \href{https://doi.org/10.48550/arXiv.2010.14313}{arXiv:2010.14313}.

    \bibitem[\protect\citename{Heyting, }1930]{heyting30}
    Heyting, A. (1930). Sur la logique intuitionniste. {\it Acad. Roy. Belg. Bull. Cl. Sci.}, 16(5), 957--63. \href{https://doi.org/10.1007/BF02410606}{doi:10.1007/BF02410606}.

    \bibitem[\protect\citename{Heyting, }1931]{heyting31}
    Heyting, A. (1931). Die intuitionistische Grundlegung der Mathematik. {\it Erkenntnis}, 2, 106--15. \href{https://doi.org/10.1007/BF02028143}{doi:10.1007/BF02028143}.

    \bibitem[\protect\citename{Heyting, }1934]{heyting34}
    Heyting, A. (1934). {\it Mathematische Grundlagenforschung Intuitionismus Beweistheorie}. Heidelberg: Springer Berlin. \href{https://doi.org/10.1007/978-3-642-65617-0}{doi:10.1007/978-3-642-65617-0}.

    \bibitem[\protect\citename{Heyting, }1956]{heyting56}
    Heyting, A. (1956). {\it Intuitionism: An Introduction}. Amsterdam: North-Holland Publishing Co. \href{https://doi.org/10.2307/2268357}{doi:10.2307/2268357}.

    \bibitem[\protect\citename{Hofmann and Streicher, }1998]{hofmann98}
    Hofmann, M. and Streicher, T. (1998). The groupoid interpretation of type theory. In Giovanni Sambin and Jan M. Smith, editors, \textit{Twenty-five years of constructive type theory (Venice, 1995), volume 36 of Oxford Logic Guides}, pages 83–111. Oxford University Press, New York. \href{https://doi.org/10.1093/oso/9780198501275.003.0008}{doi:10.1093/oso/9780198501275.003.0008}.

    \bibitem[\protect\citename{Hofstra and Warren, }2013]{hofstra13}
    Hofstra, P. and Warren M. (2013). Combinatorial realizability models of type theory. \textit{Annals of Pure and Applied Logic}, 164(10):957–988. \href{https://doi.org/10.1016/j.apal.2013.05.002}{doi:10.1016/j.apal.2013.05.002}.

    \bibitem[\protect\citename{Hyland, }1982]{hyland82}
    Hyland, J. (1982). The effective topos. \textit{Studies in logic and the foundations of mathematics}, 110:165–216. \href{https://doi.org/10.1016/S0049-237X(09)70129-6}{doi:10.1016/S0049-237X(09)70129-6}.

    \bibitem[\protect\citename{Kleene, }1945]{kleene45}
    Kleene, S. C. (1945). On the interpretation of intuitionistic number theory. {\it The Journal of Symbolic Logic}, 10(4), 109--24. \href{https://doi.org/10.2307/2269016}{doi:10.2307/2269016}.

    \bibitem[\protect\citename{Kolmogorov, }1932]{kolmogorov32}
    Kolmogorov, Andrey. (1932). Ur deutung der intuitionistischen logik. {\it Mathematische Zeitschrift}, 35, 58--65. \href{https://doi.org/10.1007/BF01186549}{doi:10.1007/BF01186549}.

    \bibitem[\protect\citename{Lack, }2010]{lack10}
    Lack, S. (2010). A 2-Categories Companion. In: Baez, J., May, J. (eds) \textit{Towards Higher Categories. The IMA Volumes in Mathematics and its Applications}, vol 152. Springer, New York, NY. \href{https://doi.org/10.1007/978-1-4419-1524-5\_4}{doi:10.1007/978-1-4419-1524-5\_4}.

    \bibitem[\protect\citename{Lambek, }1995]{lambek95}
    Lambek, J. (1995). Lambek. Some aspects of categorical logic. \textit{Studies in logic and the foundations of mathematics}, 134:69–89. \href{https://doi.org/10.1016/S0049-237X(06)80039-X}{doi:10.1016/S0049-237X(06)80039-X}.

    \bibitem[\protect\citename{Lietz and Streicher, }2002]{lietz02}
    Lietz, P. and Streicher, T. Impredicativity entails untypedness. \textit{Mathematical. Structures in Comp. Sci.}, 12(3):335–347. \href{https://doi.org/10.1017/S0960129502003663}{doi:10.1017/S0960129502003663}.

    \bibitem[\protect\citename{Longley, }1999]{longley99}
    Longley, J. (1999). Unifying typed and untyped realizability. Available online: \url{https://homepages.inf.ed.ac.uk/jrl/Research/unifying.txt}.

    \bibitem[\protect\citename{Lumsdaine, }2011]{lumsdaine11}
    Lumsdaine, P. L. (2011). A small observation on co-categories. \textit{Theory and Applications of Categories}, 25(9):247–250.

    \bibitem[\protect\citename{Mac Lane, }1971]{maclanecwm}
    Mac Lane, S. (1971). {\it Categories for the Working Mathematician}. Springer-Verlag, New York. Graduate Texts in Mathematics, Vol. 5.

    \bibitem[\protect\citename{Robinson and Rosolini, }1990]{robinson90}
    Robinson, E. and Rosolini, P. (1990). Colimit completions and the effective topos. \textit{The Journal of Symbolic Logic}, 55(2):678–699. \href{https://doi.org/10.2307/2274658}{doi:10.2307/2274658}.

    \bibitem[\protect\citename{Robinson and Rosolini, }2001]{robinson01}
    Robinson, E. and Rosolini, P. (2001). An abstract look at realizability. \textit{In International Workshop on Computer Science Logic}, pages 173–187. Springer. \href{https://doi.org/10.1007/3-540-44802-0\_13}{doi:10.1007/3-540-44802-0\_13}.

    \bibitem[\protect\citename{Schönfinkel, }1924]{schonfinkel24}
    Schönfinkel, M. (1924). Über die bausteine der mathematischen logik. \textit{Mathematische Annalen}, 92:305–316. \href{https://doi.org/10.1007/BF01448013}{doi:10.1007/BF01448013}.

    \bibitem[\protect\citename{Scott, }1993]{scott93}
    Scott, D. S. (1993). A type-theoretical alternative to ISWIM, CUCH, OWHY. \textit{Theoretical Computer Science}, Volume 121, Issues 1–2, Pages 411-440. \href{https://doi.org/10.1016/0304-3975(93)90095-B}{doi:10.1016/0304-3975(93)90095-B}.

    \bibitem[\protect\citename{Shulman, }2021]{shulman21}
    Shulman, M. (2021). Exact completion of a 2-category. Available online: \url{https://ncatlab.org/michaelshulman/show/exact+completion+of+a+2-category}.

    \bibitem[\protect\citename{Speight, }2023]{speight23}
    Speight, S. L. (2023). {\it Higher-Dimensional Realizability for Intensional Type Theory}. PhD thesis. University of Oxford.

    \bibitem[\protect\citename{Street,} 1976]{street76}
    Street, R. (1976). Limits indexed by category-valued 2-functors. \textit{J. Pure Appl. Algebra}, 8, 149–181. \href{https://doi.org/10.1016/0022-4049(76)90013-X}{doi:10.1016/0022-4049(76)90013-X}.

    \bibitem[\protect\citename{Swan and Uemura, }2021]{swan21}
    Swan, A. W. and Uemura, T. (2021). On church’s thesis in cubical assemblies. \textit{Mathematical Structures in Computer Science}, 31(10):1185–1204. \href{https://doi.org/10.1017/S0960129522000068}{doi:10.1017/S0960129522000068}.

    \bibitem[\protect\citename{Uemura, }2018]{uemura18}
    Uemura, T. (2018). Cubical assemblies, a univalent and impredicative universe and a failure of propositional resizing. In Peter Dybjer, José Espírito Santo, and Luís Pinto, editors, \textit{24th International Conference on Types for Proofs and Programs, TYPES 2018, June 18-21, 2018, Braga, Portugal, volume 130 of LIPIcs}, pages 7:1–7:20. Schloss Dagstuhl - Leibniz-Zentrum für Informatik. \href{https://doi.org/10.4230/LIPIcs.TYPES.2018.7}{doi:10.4230/LIPIcs.TYPES.2018.7}.

    \bibitem[\protect\citename{The Univalent Foundations Program, }2013]{hottbook13}
    The Univalent Foundations Program. \textit{Homotopy Type Theory: Univalent Foundations of Mathematics}. Institute for Advanced Study. Available online: \url{https://homotopytypetheory.org/book}.

    \bibitem[\protect\citename{van den Berg, }2018a]{vandenberg18a}
    van den Berg, B. (2018). Path categories and propositional identity types. \textit{ACM Trans. Comput. Logic}, 19(2). \href{https://doi.org/10.1145/3204492}{doi:10.1145/3204492}.

    \bibitem[\protect\citename{van den Berg, }2018b]{vandenberg18b}
    van den Berg, B. (2018). Univalent polymorphism. \textit{Ann. Pure Appl. Log.}, 171:102793. \href{https://doi.org/10.1016/j.apal.2020.102793}{doi:10.1016/j.apal.2020.102793}.

    \bibitem[\protect\citename{van den Berg and Moerdijk, }2018]{vandenbergm18}
    van den Berg, B. and Moerdijk, I. (2018). Exact completion of path categories and algebraic set theory: Part i: Exact completion of path categories. \textit{Journal of Pure and Applied Algebra}, 222(10):3137–3181. \href{https://doi.org/10.1016/j.jpaa.2017.11.017}{doi:10.1016/j.jpaa.2017.11.017}.

    \bibitem[\protect\citename{velebil, }1999]{velebil99}
    Velebil, J. (1999). Categorical generalization of a universal domain. \textit{Applied Categorical Structures}, 7(1):209–226. \href{https://doi.org/10.1023/A:1008608823476}{doi:10.1023/A:1008608823476}.

    \bibitem[\protect\citename{Warren, }2008]{warren08}
    Warren, M. (2008). \textit{Homotopy Theoretic Aspects of Constructive Type Theory}. PhD thesis, Carnegie Mellon University.

    \bibitem[\protect\citename{Warren, }2012]{warren12}
    Warren, M. A. (2012). A characterization of representable intervals. \textit{Theory and Applications of Categories}, 26(8):204–232.
    
\end{thebibliography}
\end{document}